\documentclass[11pt]{article}

\usepackage{fullpage}
\usepackage{complexity}
\usepackage{xcolor}
\usepackage[numbers]{natbib}
\usepackage{indentfirst}
\usepackage{bm}

\definecolor{weborange}{rgb}{.8,.3,.3}
\definecolor{webblue}{rgb}{0,0,.8}
\definecolor{internallinkcolor}{rgb}{0,.5,0}
\definecolor{externallinkcolor}{rgb}{0,0,.5}

\usepackage{hyperref}

\usepackage{amsmath}
\usepackage{amsthm}
\usepackage{mathtools}
\usepackage{mathrsfs}
\usepackage{braket}
\usepackage{amssymb}
\usepackage{cleveref}
\usepackage{paralist}
\usepackage[title]{appendix}
\usepackage{microtype}
\usepackage{bbm}
\usepackage[color=blue!10]{todonotes}
\usepackage{abstract}

\newtheorem{theorem}{Theorem}[]

\newtheorem{definition}[theorem]{Definition}
\newtheorem{lemma}[theorem]{Lemma}

\newtheorem{corollary}[theorem]{Corollary}
\newtheorem{claim}[theorem]{Claim}

\DeclareMathOperator*\err{err}
\DeclareMathOperator*\opt{opt}
\DeclareMathOperator*\Exp{{\bf E}}
\DeclareMathOperator*\Prob{{\bf Pr}}

\newcommand{\bool}{\left\{0,1\right\}}
\newcommand{\boolpm}{\left\{-1,1\right\}}

\newcommand{\bigo}[1]{O\left(#1\right)}
\newcommand{\numberthis}{\addtocounter{equation}{1}\tag{\theequation}}
\newcommand{\prblm}[1]{{\uppercase{\rm#1}}}
\newcommand{\MCSP}{\prblm{MCSP}}

\newcommand{\SG}{\prblm{SUM\text{-}GREATER}}
\newcommand{\MKtP}{\mathrm{MKtP}}

\newcommand{\AND}{{\sf AND}}

\newcommand{\FORMULA}{{\sf FORMULA}}

\newcommand{\XOR}{{\sf XOR}}
\newcommand{\LTF}{{\sf LTF}}
\newcommand{\PTF}{{\sf PTF}}
\newcommand{\SYM}{{\sf SYM}}
\newcommand{\GIP}{{\sf GIP}}

\newcommand{\EI}{\aleph}
\newcommand{\given}{\mid}
\newcommand{\uniform}{U}
\newcommand{\LH}[1]{#1^{\mathrm{L}}}
\newcommand{\RH}[1]{#1^{\mathrm{R}}}
\newcommand{\myapprox}[1]{\ensuremath{\stackrel{#1}{\approx}}}
\newcommand{\real}{\mathbb{R}}
\newcommand{\naturals}{\mathbb{N}}
\newcommand{\abs}[1]{\left|#1\right|}
\newcommand{\prn}[1]{\left(#1\right)}
\newcommand{\cprn}[1]{\!\left(#1\right)}

\newcommand{\csqbra}[1]{\!\left[#1\right]}
\newcommand{\brkts}[1]{\left\{#1\right\}}

\renewcommand\Set[1]{\csname Set \endcsname{\mskip-\medmuskip#1\mskip-\medmuskip}}

\allowdisplaybreaks

\author{
	\quad \quad Valentine Kabanets\thanks{School of Computing Science, Simon Fraser University, Burnaby, BC, Canada; \texttt{kabanets@cs.sfu.ca}.}\quad \quad
	\and
	Sajin Koroth\thanks{School of Computing Science, Simon Fraser University, Burnaby, BC, Canada; \texttt{sajin\_koroth@sfu.ca}.} \quad \quad
	\and
	Zhenjian Lu\thanks{School of Computing Science, Simon Fraser University, Burnaby, BC, Canada; \texttt{zhenjian\_lu@sfu.ca}.} \quad \quad \vspace{0.1cm}
	\and
	Dimitrios Myrisiotis\thanks{Department of Computing, Imperial College London, UK; \texttt{d.myrisiotis17@imperial.ac.uk}.}
	\and
	Igor C. Oliveira\thanks{Department of Computer Science, University of Warwick, UK; \texttt{igor.oliveira@warwick.ac.uk}.} \vspace{0.1cm}
}

\title{
Algorithms and Lower Bounds for de Morgan Formulas of Low-Communication Leaf Gates
}

\begin{document}
	
	\maketitle

	\vspace{-0.45cm}
	
	\begin{abstract}
	    The class $\FORMULA[s] \circ \mathcal{G}$ consists of Boolean functions computable by size-$s$ de Morgan formulas whose leaves are any Boolean functions from a class $\mathcal{G}$. We give \emph{lower bounds} and (SAT, Learning, and PRG) \emph{algorithms} for $\FORMULA[n^{1.99}]\circ \mathcal{G}$, for classes $\mathcal{G}$ of functions with \emph{low communication complexity}. Let $R^{(k)}(\mathcal{G})$ be the maximum $k$-party number-on-forehead randomized communication complexity of a function in $\mathcal{G}$. Among other results, we show that:
	    \begin{itemize}
    	    \item The Generalized Inner Product function $\mathsf{GIP}^k_n$ cannot be computed in  $\FORMULA[s]\circ \mathcal{G}$ on more than $1/2+\varepsilon$ fraction of inputs for
    	    $$
    	        s = o \! \left ( \frac{n^2}{ \left(k \cdot 4^k \cdot {R}^{(k)}(\mathcal{G}) \cdot \log (n/\varepsilon) \cdot \log (1/\varepsilon) \right)^{2}} \right).
    	    $$
    	    This significantly extends the lower bounds against bipartite formulas obtained by \citep{Tal17}. As a corollary, we get an average-case lower bound for $\mathsf{GIP}^k_n$ against $\FORMULA[n^{1.99}]\circ \PTF^{k-1}$, i.e., sub-quadratic-size de Morgan formulas with degree-$(k-1)$ PTF (polynomial threshold function) gates at the bottom. Previously, only sub-linear lower bounds were known \cite{Nis94, Vio15} for circuits with PTF gates.
    	    
    	    \item There is a PRG of seed length $n/2 + O\left( \sqrt{s} \cdot R^{(2)}(\mathcal{G}) \cdot\log(s/\varepsilon) \cdot \log (1/\varepsilon) \right)$ that $\varepsilon$-fools $\FORMULA[s] \circ \mathcal{G}$. For the special case of $\FORMULA[s] \circ \LTF$, i.e., size-$s$ formulas with LTF (linear threshold function) gates at the bottom, we get the better seed length $O\left(n^{1/2}\cdot s^{1/4}\cdot \log(n)\cdot \log(n/\varepsilon)\right)$. In particular, this  provides the first non-trivial PRG (with seed length $o(n)$) for intersections of $n$ half-spaces in the regime where $\varepsilon \leq 1/n$, complementing a recent result of \citep{OST19}.

	        \item There exists a randomized $2^{n-t}$-time $\#$SAT algorithm for $\FORMULA[s] \circ \mathcal{G}$, where
            $$
                t = \Omega\left(\frac{n}{\sqrt{s} \cdot \log^2(s)\cdot R^{(2)}(\mathcal{G})}\right)^{1/2}.
            $$ 
            In particular, this implies a nontrivial \#SAT algorithm for $\FORMULA[n^{1.99}]\circ \LTF$.

	        \item The Minimum Circuit Size Problem is not in $\FORMULA[n^{1.99}] \circ \mathsf{XOR}$; thereby making progress on hardness magnification, in connection with results from \citep{OPS19, Magnification_FOCS19}. On the algorithmic side, we show that the concept class $\FORMULA[n^{1.99}] \circ \mathsf{XOR}$ can be PAC-learned in time $2^{O(n/\log n)}$.
	    \end{itemize}
	\end{abstract}
	
	
	
	\newpage
	
	\setcounter{tocdepth}{3}
	
	\tableofcontents
	
	\newpage
	
    

	\section{Introduction}
	
	A (de Morgan) Boolean formula over $\{0,1\}$-valued input variables $x_1, \ldots, x_n$ is a binary tree whose internal nodes are labelled by AND or OR gates, and whose leaves are marked with a variable or its negation. The power of Boolean formulas has been intensively investigated since the early years of complexity theory (see, e.g.,~\citep{subbotovskaya1961realization, Nec66, khrapchenko1971method, andreev1987method, DBLP:journals/rsa/PatersonZ93, DBLP:journals/rsa/ImpagliazzoN93, Has98, Tal14, DM18}). The techniques underlying these complexity-theoretic results have also enabled algorithmic developments. These include learning algorithms \cite{Rei11b, DBLP:conf/innovations/ServedioT17}, satisfiability algorithms (cf.~\citep{DBLP:journals/eccc/Tal15}), compression algorithms \citep{DBLP:journals/cc/ChenKKSZ15}, and the construction of pseudorandom generators \citep{IMZ12} for Boolean formulas of different sizes. But despite many decades of research, the current  non-trivial algorithms and lower bounds apply only to formulas of less than cubic size, and understanding larger formulas remains a major open problem in circuit complexity.

	In many scenarios, however, understanding smaller formulas whose leaves are replaced by certain functions would also be very useful. Motivated by several recent works, we initiate a systematic study of the $\FORMULA \circ \mathcal{G}$ model, i.e., Boolean formulas whose leaves are labelled by an arbitrary function from a fixed class $\mathcal{G}$. This model unifies and generalizes a variety of models that have been previously studied in the literature:
	\begin{itemize}
	    \item [--]  Oliveira, Pich, and Santhanam~\citep{OPS19} show that obtaining a refined understanding of formulas of size $n^{1 + \varepsilon}$ over parity (XOR) gates would have significant consequences in complexity theory. Note that de Morgan formulas of size $n^{3 + \varepsilon}$ can simulate such devices. Therefore, a better understanding of the $\FORMULA \circ \mathcal{G}$ model even when $\mathcal{G} = \mathsf{XOR}$ is  \emph{necessary} before we are able to analyze super-cubic size formulas.\footnote{We remark that even a single layer of $\mathsf{XOR}$ gates can compute powerful primitives, such as error-correcting codes and hash functions.}

        
        \item [--] Tal~\citep{Tal17} obtains almost quadratic lower bounds for the model of bipartite formulas, where there is a fixed partition of the input variables into $x_1, \ldots, x_n$ and $y_1, \ldots, y_n$, and a formula leaf can compute an \emph{arbitrary} function over either $\vec{x}$ or $\vec{y}$.  This model was originally investigated by Pudl\'ak, R\"odl, and Savick\'y~\cite{PRS88}, where it was referred to as graph complexity. The model is also equivalent to {\sf PSPACE}-protocols in communication complexity (cf.~\citep{DBLP:journals/cc/GoosPW18}).
        
        \item [--] Abboud and Bringmann~\citep{DBLP:conf/icalp/AbboudB18} consider formulas where the leaves are threshold gates whose input wires can be arbitrary functions applied to either the first or the second half of the input. This extension of bipartite formulas  is  denoted by $\mathcal{F}_2$ in \citep{DBLP:conf/icalp/AbboudB18}. Their work establishes connections between faster $\mathcal{F}_2$-SAT algorithms, the complexity of problems in {\sf P} such as Longest Common Subsequence and the Fr\'echet Distance Problem, and circuit lower bounds.
        
        \item [--]  Polytopes (i.e.~intersection of half-spaces), which corresponds to $\mathcal{G}$ being the family of linear-threshold functions, and the formula contains only AND gates as internal gates. The constructing of PRGs for this model has received significant attention in the literature (see \citep{OST19} and references therein).
	\end{itemize}
	
	We obtain in a unified way several new results for the $\mathsf{FORMULA} \circ \mathcal{G}$ model, for natural classes $\mathcal{G}$ of functions which include parities, linear (and polynomial) threshold functions, and indeed many other functions of interest. In particular, we show that this perspective leads to stronger lower bounds, general satisfiability algorithms, and better pseudorandom generators for a broad class of functions. 

	\subsection{Results}
	
	We now describe in detail our main results and how they contrast to previous works. Our techniques will be discussed in Section \ref{sec:techniques}, while a few open problems are mentioned in Section \ref{sec:concluding}.
		
	We let $\FORMULA[s] \circ \mathcal{G}$ denote the set of Boolean functions computed by formulas containing at most $s$ leaves, where each leaf computes according to some function in $\mathcal{G}$. The set of parity functions and their negations will be denoted by $\mathsf{XOR}$.
	
	We use the following notation for communication complexity. For a Boolean function $f \colon \{0,1\}^{n} \to \{0,1\}$, we let $D(f)$ be the two-party deterministic communication complexity of $f$, where each party is given an input of $n/2$ bits. Similarly, for a Boolean function $g \colon \{0,1\}^{n} \to \{0,1\}$, we denote by $R^{(k)}_{\delta}(g)$ the communication cost of the best $k$-party \emph{number-on-forehead} (NOF) communication protocol that computes $g$ with probability at least $1-\delta$ on every input, where the probability is taken over the random choices of the protocol. For simplicity, we might omit the superscript $(k)$ from $R^{(k)}_{\delta}(g)$ when $k = 2$. One of our results will also consider $k$-party \emph{number-in-hand} (NIH) protocols, and this will be clearly indicated in order to avoid confusion. We always assume a canonical partition of the input coordinates in all statements involving $k$-party communication complexity, unless stated otherwise. We generalize these definitions for a class of functions $\mathcal{G}$ in the natural way. For instance, we let $R^{(k)}_{\delta}(\mathcal{G})=\max_{g\in\mathcal{G}} R^{(k)}_{\delta}(g)$.
    
	Our results refer to standard notions in the literature, but in order to fix notation, Section \ref{sec:preliminaries} formally defines communication protocols, Boolean formulas, and other notions relevant in this work. We refer to the textbooks \citep{Kushilevitz-Nisan97} and \citep{Juk12} for more information about communication complexity and Boolean formulas, respectively. To put our results into context, here we only briefly review a few known upper bounds on the communication complexity of certain classes $\mathcal{G}$.
	
    \paragraph{Parities ($\mathsf{XOR}$) and Bipartite Formulas.}
    
    Clearly, the deterministic two-party communication complexity of any parity function is at most $2$, since to agree on the output it is enough for the players to exchange the parity of their relevant input bits. Moreover, note that the bipartite formula model discussed above precisely corresponds to formulas whose leaves are computed by a two-party protocol of communication cost at most $1$.
		
	\paragraph{Halfspaces and Polynomial Threshold Functions (PTFs).}
	
	Recall that a halfspace, also known as a Linear Threshold Function (LTF), is a Boolean function of the form $\mathsf{sign} (\sum_i^n a_i \cdot x_i - b )$, where each $a_i, b \in \mathbb{R}$ and $x \in \{0,1\}^n$, and that a degree-$d$ PTF is its natural generalization where degree-$d$ monomials are allowed. It is known that if $g(x_1, \ldots, x_n)$ is a halfspace, then its randomized two-party communication complexity, namely $R^{(2)}_\delta(g)$, satisfies $R^{(2)}_\delta(g) = O(\log (n) + \log (1/\delta))$ \citep{Nis94}. On the other hand, if $g(x_1, \ldots, x_n)$ is a degree-$d$ PTF, then $R^{(d+1)}_\delta(g) = {O}\big((d \log d)(d \log n + \log(1/\delta))\big)$ \citep{Nis94, Vio15}.
	
	\paragraph{Degree-${\bm d}$ Polynomials over $\mathsf{GF}(2)$.}
	
	It is well known that a degree-$d$ $\mathsf{GF}(2)$-polynomial admits a $(d+1)$-party deterministic protocol of communication cost $d + 1$ under \emph{any} variable partition, since in the number-on-forehead model each monomial is entirely seen by some player. In particular, the Inner Product function $\mathsf{IP}_{n}(x,y) = \sum_i x_i \cdot y_i~(\mathsf{mod}\;2)$ satisfies $R^{(3)}_{1/3}(\mathsf{IP}_n) = O(1)$.
	
	\subsubsection{Lower bounds}
	
	Prior to this work, the only known lower bound against $\FORMULA \circ \mathsf{XOR}$ or bipartite formulas was the recent result of \citep{Tal17} showing that $\mathsf{IP}_n$ is hard (even on average) against nearly sub-quadratic formulas. In contrast, we obtain a significantly stronger result and establish lower bounds for different Boolean functions. We define such functions next.
	
	\paragraph{$\mathsf{GIP}^k_n$.}
	
	The Generalized Inner Product function $\GIP^k_{n}\colon \{0,1\}^{n} \to \{0,1\}$ is defined as
	\[
	    \GIP_{n}^{k}\left(x^{(1)},x^{(2)},\dots,x^{(k)}\right) = \sum_{j=1}^{n/k} \bigwedge_{i=1}^{k} x_{j}^{(i)}~(\mathsf{mod}\;2),
	\]
	where $x^{(i)}\in \{0,1\}^{n/k}$ for each $i\in [k]$.
	
	\paragraph{$\mathsf{MKtP}$.}
	
	In the Minimum Kt Problem, where $\mathsf{Kt}$ refers to Levin's time-bounded Kolmogorov complexity\footnote{For a string $x \in \{0,1\}^*$, $\mathsf{Kt}(x)$ denotes the minimum value $|M| + \log t$ taken over $M$ and $t$, where $M$ is a machine that prints $x$ when it computes for $t$ steps, and $|M|$ is the description length of $M$ according to a fixed universal machine $U$.}, we are given a string $x \in \{0,1\}^n$ and a string $1^\ell$. We accept  $(x,1^\ell)$ if and only if $\mathsf{Kt}(x) \leq \ell$. 
	
	\paragraph{$\mathsf{MCSP}$.}
	
	In the Minimum Circuit Size Problem, we are given as input the description of a Boolean function $f \colon \{0,1\}^{\log n} \to \{0,1\}$ (represented as an $n$-bit string), and a string $1^\ell$. We accept $(f,1^\ell)$ if and only the circuit complexity of $f$ is at most $\ell$.
	
	\vspace{0.1cm}
	
	\begin{theorem}[Lower bounds]\label{thm:main_lbs}
    	The following unconditional lower bounds hold:
    	\begin{itemize}
    	    \item[\emph{1.}] If $\mathsf{GIP}^k_n$ is $(1/2 + \varepsilon)$-close under the uniform distribution to a function in $\FORMULA[s] \circ \mathcal{G}$, then $$s \;=\; \Omega \! \left ( \frac{n^2}{k^2 \cdot 16^k \cdot \big ({R}^{(k)}_{\varepsilon/(2n^2)}(\mathcal{G}) + \log n \big )^2 \cdot \log^2 (1/\varepsilon)} \right ).$$
    	    
    	    \item[\emph{2.}] If $\mathsf{MKtP} \in \FORMULA[s] \circ \mathcal{G}$, then  
    	    $$
               s \;=\; \widetilde{\Omega} \! \left( \frac{n^2}{k^2\cdot16^{k}\cdot R^{(k)}_{1/3}(\mathcal{G})}\right).
    	    $$
    	    
    	    \item[\emph{3.}] If $\mathsf{MCSP} \in \FORMULA[s] \circ \mathsf{XOR}$, then $s = \widetilde{\Omega}(n^2)$, where $\widetilde{\Omega}$ hides inverse $\mathsf{polylog}(n)$ factors.
    	\end{itemize}
	\end{theorem}
	
	\vspace{0.1cm}
	
	Observe that, while \citep{Tal17} showed that the Inner Product function $\mathsf{IP}_n$ is hard against sub-quadratic bipartite formulas, Theorem \ref{thm:main_lbs} Item 1 yields lower bounds against formulas whose leaves can compute bounded-degree PTFs and $\mathsf{GF}(2)$-polynomials, including $\mathsf{IP}_n$. 
    PTF circuits were previously studied by Nisan \cite{Nis94}, who obtained an almost linear $\Omega_d(n^{1-o(1)})$ gate complexity lower bound against circuits with degree-$d$ PTF gates. Recently, \cite{KKL17} gave a super-linear \emph{wire} complexity lower bound for constant-depth circuits with constant-degree PTF gates. However, it was open whether we can prove lower bounds against any circuit model that can incorporate a linear number of PTF gates. In fact, it was open before this work to show a super-linear gate complexity lower bound against $\AND\circ \PTF$.
	
	Let us now comment on the relevance of Items 2 and 3. Both $\mathsf{MCSP}$ and $\mathsf{MKtP}$ are believed to be computationally much harder than $\mathsf{GIP}_n^k$. However, it is more difficult to analyze these problems compared to $\mathsf{GIP}_n^k$ because the latter is mathematically ``structured,'' while the former problems do not seem to be susceptible to typical algebraic, combinatorial, and analytic techniques.
	
	More interestingly, $\mathsf{MCSP}$ and $\mathsf{MKtP}$ play an important role in the theory of hardness magnification (see \citep{OPS19, Magnification_FOCS19}). In particular, if one could show that $\mathsf{MCSP}$ restricted to an input parameter $\ell \leq n^{o(1)}$ is not in $\FORMULA[n^{1 + \varepsilon}] \circ \mathsf{XOR}$ for some $\varepsilon > 0$, then it would follow that $\mathsf{NP}$ cannot be computed by Boolean formulas of size $n^c$, where $c \in \mathbb{N}$ is arbitrary. Theorem \ref{thm:main_lbs} makes partial progress on this direction by establishing the first lower bounds for these problems in the $\FORMULA \circ \mathcal{G}$ model. (We note that the proof of Theorem \ref{thm:main_lbs} Item 3 requires instances where the parameter $\ell$ is $n^{\Omega(1)}$.)
	
	\subsubsection{Pseudorandom generators}\label{sec:intro_prg}
	
	We also get pseudorandom generators (PRGs) against $\FORMULA\circ\mathcal{G}$ for various classes of functions $\mathcal{G}$. Recall that a PRG against a class of functions $\mathfrak{C}$ is a function $G$ mapping short Boolean strings (seeds) to longer Boolean strings, so that every function in $\mathfrak{C}$ accepts $G$'s output on a uniformly random seed with about the same probability as that for an actual uniformly random string. More formally, $G\colon\{0,1\}^{\ell}\to\{0,1\}^n$ is a  PRG that $\varepsilon$-fools $\mathfrak{C}$ if for every Boolean function $h \colon \{0,1\}^n \to \{0,1\}$ in $\mathfrak{C}$, we have
	\[
        \left|\Prob_{z\in\{0,1\}^{\ell}}[h(G(z))=1] - \Prob_{x\in\{0,1\}^n}[h(x)=1]\right| \;\leq\; \varepsilon.
	\]
	Furthermore, we require $G$ to run in deterministic time $\mathsf{poly}(n)$ on an input string $z \in \{0,1\}^\ell$. The parameter $\ell = \ell(n)$ is called the seed length of the PRG and is the main quantity to be minimized when constructing PRGs.

	There exists a PRG that fools formulas of size $s$ and that has a seed of length $s^{1/3 + o(1)}$~\citep{IMZ12}. In particular, there are non-trivial PRGs for $n$-variate formulas of size nearly $n^3$. Unfortunately, such PRGs cannot be used to fool even linear size formulas over parity functions, since the naive simulation of these enhanced formulas by standard Boolean formulas requires size $n^3$. Moreover, it is not hard to see that this simulation is optimal: Andreev's function, which is hard against formulas of nearly cubic size (cf.~\citep{Has98}), can be easily computed in $\FORMULA[O(n)] \circ \mathsf{XOR}$. Given that a crucial idea in the construction of the PRG in \citep{IMZ12} (shrinkage under restrictions) comes from this lower bound proof, new techniques are needed in order to approach the problem in the $\FORMULA \circ \mathsf{XOR}$ model.
	
	More generally, extending a computational model for which strong PRGs are known to allow parities at the bottom layer can cause significant difficulties. A well-known example is $\mathsf{AC}^0$ circuits and their extension to $\mathsf{AC}^0$-$\mathsf{XOR}$. While the former class admits PRGs of poly-logarithmic seed length (see e.g.~\citep{ST19}), the most efficient PRG construction for the latter has seed length $(1 - o(1)) \cdot n$ \citep{DBLP:journals/toc/FeffermanSUV13}. Consequently, designing PRGs of seed length $\leq (1 - \Omega(1)) \cdot n$ can already be a challenge. We are not aware of previous results on PRGs for $\FORMULA \circ \mathcal{G}$ for any non-trivial class $\mathcal{G}$. 
	
	By combining ideas from circuit complexity and communication complexity, we construct PRGs of various seed lengths for $\FORMULA\circ\mathcal{G}$, where $\mathcal{G}$ ranges from the class of parity functions to the much larger class of functions of bounded randomized $k$-party communication complexity.
	
	\begin{theorem}[Pseudorandom generators] \label{thm:main_PRG}
    	Let $\mathcal{G}$ be a class of $n$-bits functions. Then,
    	\begin{itemize}
            \item[\emph{1.}] In the context of parity functions, there is a \emph{PRG} that $\varepsilon$-fools $\FORMULA[s] \circ \mathsf{XOR}$ of seed length
            $$
                \ell \; = \; O \! \left (\sqrt{s}\cdot\log(s)\cdot\log(1/\varepsilon) + \log(n)\right).
            $$
            
            \item[\emph{2.}] In the context of two-party randomized communication complexity, there is a \emph{PRG} that $\varepsilon$-fools $\FORMULA[s] \circ \mathcal{G}$ of seed length
            $$
                \ell \; = \;  n/2+O \! \left(\sqrt{s}\cdot\left(R^{(2)}_{\varepsilon/(6s)}(\mathcal{G})+\log(s)\right)\cdot\log(1/\varepsilon)\right).
            $$
            More generally, for every $k(n) \geq 2$, let $\mathcal{G}$ be the class of functions that have $k$-party number-in-hand \emph{(NIH)} $(\varepsilon/6s)$-error randomized communication protocols of cost at most $R^{(k\text{-}\mathsf{NIH})}_{\varepsilon/(6s)}$. There exists a \emph{PRG} that $\varepsilon$-fools $\FORMULA[s]\circ \mathcal{G}$ with seed length
            $$
               \ell \; = \; n/k+O \! \left(\sqrt{s}\cdot \left (R^{(k\text{-}\mathsf{NIH})}_{\varepsilon/(6s)} +\log (s) \right )\cdot \log(1/\varepsilon) +\log(k)\right)\cdot\log(k).
            $$
            
            \item[\emph{3.}] In the setting of $k$-party \emph{NOF} randomized communication complexity, there is a \emph{PRG} that $\varepsilon$-fools $\FORMULA[s] \circ \mathcal{G}$ of seed length
            $$
                \ell \; = \; n-\frac{n}{O \! \left(\sqrt{s}\cdot k \cdot 4^k \cdot \left(  R^{(k)}_{\varepsilon/(2s)}(\mathcal{G})+\log(n)\right)\cdot\log(n/\varepsilon)\right)}.
            $$
    	\end{itemize}
	\end{theorem}
	
    A few comments are in order. Under a standard connection between PRGs and lower bounds (see e.g.~\cite{kabanets2002derandomization}), improving the dependence on $s$ in the seed length for $\FORMULA[s] \circ \mathsf{XOR}$ (Theorem \ref{thm:main_PRG} Item 1) would require the proof of super-quadratic lower bounds against $\FORMULA \circ \mathsf{XOR}$. We discuss this problem in more detail in Section \ref{sec:concluding}. Note that the additive term $n/2$ is necessary in Theorem \ref{thm:main_PRG} Item 2, since the model computes in particular every Boolean function on the first $n/2$ input variables (i.e.~a protocol of communication cost $1$). Similarly, $\ell \geq (1 - 1/k) \cdot n$ in Theorem \ref{thm:main_PRG} Item 3. Removing the exponential dependence on $k$ would also require advances in state-of-the-art lower bounds for multiparty communication complexity.

	Theorem \ref{thm:main_PRG} Item 2 has an interesting implication for fooling a well-studied class of functions: \emph{intersections of halfspaces}.\footnote{Clearly, the intersection of $s$ functions can be computed by an enhanced formula of size $s + 1$.} Note that an intersection of halfspaces is precisely a polytope, or equivalently, the set of solutions of a $0$-$1$ integer linear program. Such objects have found applications in many fields, including optimization and high-dimensional geometry.  After a long sequence of works on the construction of PRGs for bounded-weight halfspaces, (unrestricted) halfspaces, and generalizations of these classes,\footnote{We refer to the recent reference \citep{OST19} for an extensive review of the literature in this area.}  the following results are known for the intersection of $m$ halfspaces over $n$ input variables. Gopalan, O'Donnell, Wu, and Zuckerman~\citep{GOWZ10} gave a PRG for this class for error $\varepsilon$  with seed length 
	\[
	    O\big (m\cdot\log(m/\varepsilon)+\log n)\cdot\log(m/\varepsilon) \big).
    \]
    Note that the seed length of their PRG becomes trivial if the number of halfspaces is linear in $n$. More recently, O'Donnell, Servedio and Tan~\cite{OST19} constructed a PRG with seed length 
    \[
        \poly(\log(m),1/\varepsilon)\cdot\log(n).
    \]
    Their PRG has a much better dependence on $m$, but it cannot be used in the small error regime. For example, the seed length becomes trivial if $\varepsilon = 1/n$. In particular, before this work it was open to construct a non-trivial PRG for the following natural setting of parameters (cf.~\citep[Section 1.2]{OST19}): intersection of $n$ halfspaces with error $\varepsilon = 1/n$. 
    
    We obtain the following consequence of Theorem \ref{thm:main_PRG} Item 2, which follows from a result of Viola \citep{Vio15} on the $k$-party \emph{number-in-hand} randomized communication complexity of a halfspace.
    
    \begin{corollary}[Fooling intersections of halfspaces in the low-error regime]\label{cor:PRG_LTF}
        For every $n, m \in \mathbb{N}$ and $\varepsilon > 0$, there is a pseudorandom generator with seed length 
        \[
			O\! \left(n^{1/2}\cdot m^{1/4}\cdot \log(n)\cdot \log(n/\varepsilon)\right).
		\]
		that $\varepsilon$-fools the class of intersections of $m$ halfspaces over $\{0,1\}^n$.
    \end{corollary}
    
    We note that the PRG from Theorem \ref{thm:main_PRG} Item 3 can fool, even in the exponentially small error regime, not only intersections of halfspaces, but also small formulas over bounded-degree PTFs. 
    
    Finally, Theorem \ref{thm:main_PRG} Item 2 yields the first non-trivial PRG for formulas over symmetric functions. Let $\mathsf{SYM}$ denote the class of symmetric Boolean functions on any number of input variables.
    
    \begin{corollary}[Fooling sub-quadratic formulas over symmetric gates]\label{cor:PRG_FORM_SYM}
        For every $n, s \in \mathbb{N}$ and $\varepsilon > 0$, there is a pseudorandom generator with seed length
        \[
		    O \! \left(n^{1/2}\cdot s^{1/4}\cdot \log(n)\cdot \log(1/\varepsilon) \right).
		\]
		that $\varepsilon$-fools $n$-variate Boolean functions in $\FORMULA[s]\circ \SYM$. 
	\end{corollary}
	
	Prior to this work, Chen and Wang \citep{DBLP:conf/innovations/ChenW19} proved that the number of satisfying assignments of an $n$-variate formula of size $s$ over symmetric gates can be approximately counted to an additive error term $\leq \varepsilon\cdot 2^n$ in deterministic time $\exp{\!(n^{1/2}\cdot s^{1/4 + o(1)} \sqrt{(\log (n) + \log (s))} )}$, where $\varepsilon >0$ is an arbitrary constant. While their upper bound is achieved by a white-box algorithm, Corollary \ref{cor:PRG_FORM_SYM} provides a (black-box) PRG for the same task.
	
	
	
	\subsubsection{Satisfiability algorithms}
	
	In the $\#$SAT problem for a computational model $\mathcal{C}$, we are given as input the description of a computational device $D(x_1, \ldots, x_n)$ from $\mathcal{C}$, and the goal is to count the number of satisfying assignments for $D$. This generalizes the SAT problem for $\mathcal{C}$, where it is sufficient to decide whether $D$ is satisfiable by some assignment. 
	
	In this section, we show that $\#$SAT algorithms can be designed for a broad class of functions. We consider the $\FORMULA \circ \mathcal{G}$ model for classes $\mathcal{G}$ that admit two-party communication protocols of bounded cost. We establish a general result in this context which can be used to obtain algorithms for previously studied classes of Boolean circuits.
	
	To put our $\#$SAT algorithms for $\FORMULA \circ \mathcal{G}$ into context, we first mention relevant related work on the satisfiability of Boolean formulas. Recall that in the very restricted setting of CNF formulas, known algorithms run (in the worst-case) in time $2^{n - o(n)}$ when the input formulas can have a super-linear number of clauses (cf.~\citep{dantsin2009worst}). On the other hand, for the class of general formulas, there is a better-than-brute-force algorithm for formulas of size almost $n^3$. In more detail, for any $\varepsilon >0$, there is a deterministic $\#$SAT algorithm for $\FORMULA[n^{3 - \varepsilon}]$ that runs in time $2^{n - n^{\Omega(\varepsilon)}}$ \citep{DBLP:journals/eccc/Tal15}. No results are known for formulas of cubic size and beyond, and for the reasons explained in Section \ref{sec:intro_prg}, the algorithm from \citep{DBLP:journals/eccc/Tal15} cannot even be applied to $\FORMULA \circ \mathsf{XOR}$.  
	
	Before stating our results, we discuss the input encoding in the $\#$SAT problem for $\FORMULA \circ \mathcal{G}$. The top formula $F$ is represented in some canonical way, while for each leaf $\ell$ of $F$, the input string contains the description of a protocol $\Pi_\ell$ computing a function in $\mathcal{G}$. Our results are robust to the encoding employed for $\Pi_\ell$. Recall that a protocol for a two-party function is specified by a protocol tree and a sequence of functions, where each function is associated with some internal node of the tree and depends on $n/2$ input bits. Since a protocol of communication cost $o(n)$ has a protocol tree containing at most $2^{o(n)}$ nodes, it can be specified by a string of length $2^{n/2 + o(n)}$. Our algorithms will run in time closer to $2^n$, and using a fully explicit input representation for the protocols is not an issue. Another possibility for the input representation is to use ``computational efficient'' protocols. Informally, the next bit messages of such protocols can be computed in polynomial time from the current transcript of the protocol and a player input. An advantage of this representation is that an input to our $\#$SAT problem can be succinctly represented. We observe that these input representations can be generalized to randomized two-party protocols in natural ways. We refer to Section \ref{sec:preliminaries} for a formal presentation.
	
	We obtain non-trivial satisfiability algorithms assuming upper bounds on the two-party deterministic and randomized communication complexities of functions in $\mathcal{G}$.
	
	\begin{theorem}[Satisfiability algorithms]\label{thm:main_sat}
	The following results hold.
    	\begin{itemize}
    	    \item[\emph{1.}] There is a deterministic $\#$\emph{SAT} algorithm for $\FORMULA[s] \circ \mathcal{G}$ that runs in time
    	    $$
    	        2^{n - t},~\text{where}~t = \Omega\! \left ( \frac{n} {\sqrt{s} \cdot \log^2(s) \cdot D(\mathcal{G}) }\right).
    	    $$
    	    
    	    \item[\emph{2.}] There is a randomized $\#$\emph{SAT} algorithm for $\FORMULA[s] \circ \mathcal{G}$ that runs in time
    	    $$
    	        2^{n - t},~\text{where}~t = \Omega\! \left ( \frac{n} {\sqrt{s} \cdot \log^2(s) \cdot R_{1/3}(\mathcal{G}) }  \right )^{\! 1/2}.
    	    $$
    	\end{itemize}
	\end{theorem}
	
	Theorem \ref{thm:main_sat} readily provides algorithms for many circuit classes. For instance, since one can effectively describe a randomized communication protocol for linear threshold functions \citep{Nis94,Vio15}, the algorithm from Theorem \ref{thm:main_sat} Item 2 can be used to count the number of satisfying assignments of Boolean devices from $\FORMULA[n^{1.99}] \circ \mathsf{LTF}$.
	
	\begin{corollary}[{$\#$SAT algorithm for formulas of linear threshold functions}]\label{cor:sat_formula_ltf}
    	There is a randomized $\#$\emph{SAT} algorithm for $\FORMULA[s] \circ \mathsf{LTF}$ that runs in time 
    	$$
            2^{n - t},~\text{where}~t = \Omega\! \left ( \frac{n} {\sqrt{s} \cdot \log^2(s) \cdot \log(n)}  \right )^{\! 1/2}.
        $$
	\end{corollary}
	
	In connection with Corollary \ref{cor:sat_formula_ltf}, prior to this work essentially two lines of research have been pursued. $\#$SAT and/or SAT algorithms were known for bounded-depth circuits of almost-linear size whose gates can compute LTFs or sparse PTFs (see \citep{DBLP:conf/approx/KabanetsL18} and references therein), and for sub-exponential size $\mathsf{ACC}^0$ circuits with two layers of LTFs at the bottom, assuming a sub-quadratic number of them in the layer next to the input variables (see \citep{ACW16} for this result and further related work). Corollary \ref{cor:sat_formula_ltf} seems to provide the first non-trivial SAT algorithm that operates with unbounded-depth Boolean devices containing a layer with a sub-quadratic number of LTFs.
	
	Theorem \ref{thm:main_sat} can be seen as a generalization of several approaches to designing SAT algorithms appearing in the literature, which often employ ad-hoc constructions to convert bottlenecks in the computation of devices from a class $\mathcal{C}$ into non-trivial SAT algorithms for $\mathcal{C}$. We observe that, before this work, \citep{DBLP:conf/soda/PatrascuW10} had made a connection between faster SAT algorithms for CNFs and the 3-party communication complexity of a specific function. Their setting is different though: it seems to work only for CNFs, and they rely on conjectured upper bounds on the communication complexity of a particular problem. More recently, \citep{DBLP:conf/innovations/ChenW19} employed quantum communication protocols to design \emph{approximate counting} algorithms for several problems.\footnote{Recall that approximately counting satisfying assignments is substantially easier than solving $\#$SAT, for which the fastest known algorithms run in time $2^{(1-o(1))n}$.} In comparison to previous works, to our knowledge  Theorem \ref{thm:main_sat} is the first unconditional result that yields faster $\#$SAT algorithms via communication complexity in a generic way.%
    \footnote{It has been brought to our attention that Avishay Tal has independently discovered a SAT algorithm for bipartite formulas of sub-quadratic size (see the discussion in \citep[Page 7]{DBLP:conf/icalp/AbboudB18}), which corresponds to a particular case of Theorem \ref{thm:main_sat}.}
	
	\subsubsection{Learning algorithms}
	
	We describe a learning algorithm for the $\FORMULA \circ \mathsf{XOR}$ class in Leslie Valiant's challenging PAC-learning model \citep{Val84}. Recall that a (PAC) learning algorithm for a class of functions $\mathcal{C}$ has access to labelled examples $(x,f(x))$ from an unknown function $f \in \mathcal{C}$, where $x$ is sampled according to some (also unknown) distribution $\mathcal{D}$. The goal of the learner is to output, with high probability over its internal randomness and over the choice of random examples (measured by a confidence parameter $\delta$), a hypothesis $h$ that is close to $f$ under $\mathcal{D}$ (measured by an error parameter $\varepsilon$). We refer to \citep{KV94} for more information about this learning model, and to Section \ref{sec:preliminaries} for its standard formalization.
	
	It is known that formulas of size $s$ can be PAC-learned in time $2^{\widetilde{O}(\sqrt{s})}$ \citep{Rei11b}. Therefore, formulas of almost quadratic size can be non-trivially learned from random samples of an arbitrary distribution. A bit more formally, we say that a learning algorithm is \emph{non-trivial} if it runs in time $2^n/n^{\omega(1)}$, i.e., noticeably faster than the trivial brute-force algorithm that takes time $2^n \cdot \mathsf{poly}(n)$. Obtaining non-trivial learning algorithms for various circuit classes is closely connected to the problem of proving explicit lower bounds against the class \citep{OS17} (see also \citep{DBLP:conf/innovations/ServedioT17} for a systematic investigation of such algorithms). We are not aware of the existence of non-trivial learning algorithms for super-quadratic size formulas. However, it seems likely that such algorithms exist at least for formulas of near cubic size. As explained in Section \ref{sec:intro_prg}, this would still be insufficient for the learnability of classes such as (linear size) $\FORMULA \circ \mathsf{XOR}$.
	
	
	We explore structural properties of $\FORMULA \circ \mathsf{XOR}$ employed in previous results and boosting techniques from learning theory to show that sub-quadratic size devices from this class can be PAC-learned in time $2^{O(n/\log n)}$.
	
	\begin{theorem}[PAC-learning $\FORMULA \circ \mathsf{XOR}$ in sub-exponential time]\label{thm:main_learning} 
    	For every constant $\gamma > 0$, there is an algorithm that \emph{PAC} learns the class of $n$-variate Boolean functions $\FORMULA[n^{2 - \gamma}]\circ\mathsf{XOR}$ to accuracy $\varepsilon$ and with confidence $\delta$ in time $\mathsf{poly}\big (2^{n/\log n}, 1/\varepsilon, \log(1/\delta) \big )$.
	\end{theorem}
	
	
	Note that a sub-exponential running time cannot be achieved for $\FORMULA \circ \mathcal{G}$ when we consider the communication complexity of $\mathcal{G}$. Again, the class is too large, for the same reason discussed in Section \ref{sec:intro_prg}. It might still be possible to design a non-trivial learning algorithm in this case, but this would possibly require the introduction of new lower bound techniques for $\FORMULA \circ \mathsf{XOR}$. 
	
    In contrast to the algorithm mentioned above that learns (standard) formulas of size $s \leq n^{2 - o(1)}$ in time $2^{\widetilde{O}(\sqrt{s})}$, the algorithm from Theorem \ref{thm:main_learning} does not learn smaller formulas over parities in time faster than $2^{O(n/\log n)}$. We discuss this in more detail in Sections \ref{sec:techniques} and \ref{sec:concluding}.
    
    Finally, we mention a connection to cryptography that provides a conditional upper bound on the size of $\FORMULA \circ \mathsf{XOR}$ circuits that can be learned in time $2^{o(n)}$. It is well known that if a circuit class $\mathcal{C}$ can compute pseudorandom functions (or some variants of this notion), then it cannot be learned in various learning models (see e.g.~\citep{KV94}). It has been recently conjectured that depth-two $\mathsf{MOD}_3 \circ \mathsf{XOR}$ circuits of linear size can compute weak pseudorandom functions of exponential security \citep[Conjecture 3.7]{BIP+18}. If this conjecture holds, then such circuits cannot be learned in time $2^{o(n)}$. Since $\mathsf{MOD}_3$ gates over a linear number of input wires can be simulated by formulas of size at most $O(n^{2.8})$ \citep{sergeev2017upper}, under this cryptographic assumption it is not possible to learn $\FORMULA[n^{2.8}] \circ \mathsf{XOR}$ in time $2^{o(n)}$, even if the learner only needs to succeed under the uniform distribution.
	
	\subsection{Techniques}\label{sec:techniques}

    In order to explain our techniques, we focus for the most part on the design of PRGs for $\FORMULA \circ \mathcal{G}$ when $\mathcal{G}$ is of bounded two-party randomized communication complexity (a particular case of Theorem \ref{thm:main_PRG} Item 2). This proof makes use of various ingredients employed in other results. After sketching this argument, we say a few words about our strongest lower bound (Theorem \ref{thm:main_lbs} Item 1) and the satisfiability and learning algorithms (Theorems \ref{thm:main_sat} and \ref{thm:main_learning}, respectively).
    
    We build on a powerful result showing that any small de Morgan formula can be  approximated pointwise by a low-degree polynomial:
    
    \vspace{0.25cm}
    
    \noindent\textbf{(A)} For every formula $F(y_1, \ldots, y_m)$ of size $s$, there is a polynomial $p(y_1, \ldots, y_m) \in \mathbb{R}[y_1, \ldots, y_m]$ of degree $O(\sqrt{s})$ such that $|F(a) - p(a)| \leq 1/10$ on every $a \in \{0,1\}^m$.
    
    \vspace{0.25cm}
    
    \noindent The only known proof of this result \citep{Rei11b} relies on a sequence of works \citep{BBC+01, DBLP:journals/cc/LaplanteLS06, DBLP:conf/stoc/HoyerLS07, DBLP:journals/toc/FarhiGG08, DBLP:conf/focs/Reichardt09, DBLP:journals/siamcomp/AmbainisCRSZ10, DBLP:journals/toc/ReichardtS12} on quantum query complexity, generalizing Grover's search algorithm for the OR predicate \citep{DBLP:conf/stoc/Grover96} to arbitrary formulas. The starting point of many of our results is a consequence of \textbf{(A)} which is implicit in the work of Tal \citep{Tal17}.
    
    \vspace{0.25cm}
    
    \noindent\textbf{(B)} Let $\mathcal{D}$ be a distribution over $\{0,1\}^m$, and $F \in \FORMULA[s] \circ \mathcal{G}$. Then, for every function $f$,
    \vspace{-0.1cm}
    $$
        \text{if~} \Pr_{x\sim \mathcal{D}}[F(x) = f(x)] \geq 1/2 + \varepsilon\text{~~then~~} \Pr_{x \sim \mathcal{D}}[h(x) = f(x)] \geq 1/2 + \exp{\!(-t)}
        \vspace{-0.15cm}
    $$
    for some function $h$ which is the XOR of at most $t$ functions in $\mathcal{G}$, where $t = \widetilde{\Theta}(\sqrt{s} \cdot \log(1/\varepsilon))$.
    
    \vspace{0.25cm}
    
    \noindent Intuitively, if we could understand well enough the XOR of any small collection of functions in $\mathcal{G}$, then we can translate this into results for $\FORMULA[s] \circ \mathcal{G}$, as long as $s \ll n^2$. We adapt the techniques behind $\textbf{(B)}$ to provide a general approach to constructing PRGs against $\FORMULA \circ \mathcal{G}$:
    
    \vspace{0.25cm}
    
    \noindent\textbf{Main PRG Lemma.} In order for a distribution $\mathcal{D}$ to $\varepsilon$-fool the class $\FORMULA[s] \circ \mathcal{G}$, it is enough for it to $\exp{\!(-t)}$-fool the class $\mathsf{XOR}_t \cdot \mathcal{G}$, where $t = \widetilde{\Theta}(\sqrt{s} \cdot \log(1/\varepsilon))$.
    
    \vspace{0.25cm}
    
    \noindent Recall that, in Theorem \ref{thm:main_PRG} Item 2, we consider a class $\mathcal{G}$ of functions that admit two-party randomized protocols of cost $R = R^{(2)}_{\varepsilon/6s}(\mathcal{G})$. It is easy to see that the XOR of any $t$ functions from $\mathcal{G}$ is a function that can be computed by a protocol of cost at most $t \cdot R$. Thus the lemma above shows that it is sufficient to fool, to exponentially small error, a class of functions of bounded two-party randomized communication complexity. Moreover, since a randomized protocol can be written as a convex combination of deterministic protocols, it is possible to prove that fooling functions of bounded deterministic communication complexity is enough. 
    
    Pseudorandom generators in the two-party communication model have been known since \citep{INW94}. Their construction exploits that the Boolean matrix associated with a function of small communication cost can be partitioned into a not too large number of monochromatic rectangles. We provide in Appendix \ref{appendix:prg_cc} a slightly modified and self-contained construction based on explicit extractors. It achieves the following parameters: There is an explicit PRG that $\delta$-fools any $n$-bit function of two-party communication cost $D$ and that has seed length $n/2 + O(D + \log(1/\delta))$. This PRG has non-trivial seed length even when the error is exponentially small, as required by our techniques. 
    One issue here is that the INW PRG was only shown to fool functions with low \emph{deterministic} communication complexity. To obtain our PRGs for $\FORMULA\circ\mathcal{G}$ when $\mathcal{G}$ admits low-cost \emph{randomized} protocols, we first extend the analysis of the INW PRG to show that it also fools functions with low \emph{randomized} communication complexity.
    Combining this construction with the aforementioned discussion completes the proof of Theorem \ref{thm:main_PRG} Item 2.
    
    The argument just sketched reduces the construction of PRGs for $\FORMULA \circ \mathcal{G}$ when functions in $\mathcal{G}$ admit low-cost \emph{randomized} protocols to the analysis of PRGs for functions that admit relatively low-cost \emph{deterministic} protocols. Our lower bound proof for $\mathsf{GIP}^k_n$ in Theorem \ref{thm:main_lbs} Item 1 proceeds in a similar fashion. We combine statement \textbf{(B)} described above with other ideas to show:
    
    \vspace{0.25cm}
    
    \noindent\textbf{Transfer Lemma (Informal).} If a function correlates with some small formula whose leaf gates have low-cost \emph{randomized} $k$-party protocols, then it also non-trivially correlates with some function that has relatively low-cost \emph{deterministic} $k$-party protocols.
    
    \vspace{0.25cm}
    
    \noindent Given this result, we are able to rely on a strong average-case lower bound for $\mathsf{GIP}^k_n$ against $k$-party deterministic protocols from \citep{BNS92} to conclude that $\mathsf{GIP}^k_n$ is hard for $\FORMULA \circ \mathcal{G}$.
    
    Our $\#$SAT algorithms combine the polynomial representation of the top formula provided by \textbf{(A)}, for which we show that such a polynomial can be obtained \emph{explicitly}, with a decomposition of the Boolean matrix at each leaf that is induced by a corresponding low-cost randomized or deterministic two-party protocol. A careful combination of these two representations allows us to adapt a standard technique employed in the design of non-trivial SAT algorithms (fast rectangular matrix multiplication) to obtain non-trivial savings in the running time. 
    
    Finally, our learning algorithm for $\mathsf{FORMULA} \circ \mathsf{XOR}$ is a consequence of statement \textbf{(B)} above coupled with standard tools from learning theory. In a bit more detail, since a parity of parities is just another parity function, \textbf{(B)} implies that, under any distribution, every function in $\mathsf{FORMULA}[n^{1.99}] \circ \mathsf{XOR}$ is weakly correlated with some parity function. Using the agnostic learning algorithm for parity functions of \citep{KMV08}, it is possible to weakly learn $\mathsf{FORMULA}[n^{1.99}] \circ \mathsf{XOR}$ in time $2^{O(n/\log n)}$. This weak learner can then be transformed into a (strong) PAC learner using standard boosting techniques \citep{Fre90}, with only a polynomial blow-up over its running time.

	\subsection{Concluding remarks}\label{sec:concluding}

    
    
    The main message of our results is that the \emph{computational power} of a subquadratic-size top formula is \emph{not} significantly enhanced by leaf gates of \emph{low communication complexity}. We believe that the idea of decomposing a Boolean device into a computational part and a layer of communication protocols will find further applications in lower bound proofs and algorithm design.
    
    One of our main open problems is to discover a method that can analyze $\FORMULA[s] \circ \mathcal{G}$ when $s \gg n^2$. For instance, is it possible to adapt existing techniques to show an explicit lower bound against $\FORMULA[n^{2.01}] \circ \mathcal{G}$, or achieving this is just as hard as breaking the cubic barrier for formula lower bounds? Results in this direction would be interesting even for $\mathcal{G} = \mathsf{XOR}$.
    
    Finally, we would like to mention a few questions connected to our results and their applications. Is it possible to combine the techniques behind Corollary \ref{cor:PRG_LTF} and~\citep{OST19} to design a PRG of seed length $n^{o(1)}$ and error $\varepsilon = 1/n$  for the intersection of $n$ halfspaces? Can we design a satisfiability algorithm for formulas over $k$-party number-on-forehead communication protocols? Is it possible to learn $\FORMULA[s] \circ \mathsf{XOR}$ in time $2^{\widetilde{O}(\sqrt{s})}$? (The learning algorithm for formulas from~\citep{Rei11b} relies on techniques from~\citep{DBLP:journals/siamcomp/KalaiKMS08}, and it is unclear how to extend them to the case of $\FORMULA \circ \mathsf{XOR}$.)
 
        
        
    
    \subsection{Organization}
    
    Theorem \ref{thm:main_lbs} Item 1 is proved in Section \ref{sec:lower-bounds}, while Items 2 and 3 rely on our PRG constructions and are deferred to Section \ref{sec:prgs}. The latter describes a general approach to constructing PRGs for $\FORMULA \circ \mathcal{G}$. It includes the proof of Theorem \ref{thm:main_PRG} and other applications. Our satisfiability algorithms (Theorem \ref{thm:main_sat}) appear in Section \ref{sec:sat}. Finally, Section \ref{sec:learning} discusses learning results for $\FORMULA \circ \mathsf{XOR}$ and contains a proof of Theorem \ref{thm:main_learning}.

	\section{Preliminaries}\label{sec:preliminaries}
	
	\subsection{Notation}
	
	Let $n\in\mathbb{N}$; we denote $\brkts{1,\dots,n}$ by $[n]$, and denote by $\uniform_n$ the uniform distribution over $\bool^n$. We use $\widetilde{O}(\cdot)$ (and $\widetilde{\Omega}(\cdot)$)  to hide polylogarithmic factors. That is, for any $f\colon\mathbb{N}\to\mathbb{N}$, we have that $\widetilde{O}(f(n))=O\cprn{f(n)\cdot\polylog\cprn{f(n)}}$.
	
	In this paper, we will mainly use $\boolpm$ as the Boolean basis. In some parts of this paper, we will use the $\bool$ basis for the simplicity of the presentation. This will be specified in corresponding sections.
	
	\subsection{De Morgan formulas and extensions}
	
	\begin{definition}
	    An $n$-variate \emph{de Morgan formula} is a directed rooted tree; its non-leaf vertices (henceforth, \emph{internal gates}) take labels from $\brkts{{\rm AND},{\rm OR},{\rm NOT}}=\brkts{\land,\lor,\neg}$ and its leaves (henceforth, \emph{variable gates}) take labels from the set of variables $\brkts{x_1,\dots,x_n}$. Each internal gate has bounded in-degree (henceforth, \emph{fan-in}); the {\rm NOT} gate in particular has fan-in $1$ and every variable gate has fan-in $0$. The \emph{size} of a de Morgan formula is the number of its leaf gates.
	\end{definition}
	
	In this work, we denote by $\FORMULA[s]$ the class of Boolean functions computable by size-$s$ de Morgan formulas. Let $\mathcal{G}$ denote some class of Boolean functions; then, we denote by $\FORMULA[s]\circ\mathcal{G}$ the class of functions computable by some size-$s$ de Morgan formula where its leaves are labelled by functions in $\mathcal{G}$.
	
	
    
    
    \subsection{Approximating polynomials}
    
    \begin{definition}[Point-wise approximation]
        For a Boolean function $f\colon\boolpm^n\to\boolpm$, we say that the function $\tilde{f}\colon\boolpm^n\to\real$ $\varepsilon$-approximates $f$ if for every $z\in\boolpm^n$,
		\[
			\left| f(z)-\tilde{f}(z)\right| \leq \varepsilon.
		\]
    \end{definition}
    
    We will need the following powerful result for the approximating degree of de Morgan formulas.
    
    \begin{theorem}[\cite{Rei11b}, see also~\cite{ BNRdW07}]\label{thm:approx-polynomial}
		Let $s>0$ be an integer and $0<\varepsilon<1$. Any de Morgan formula $F\colon\boolpm^n\to\boolpm$ of size $s$ has a $\varepsilon$-approximating polynomial of degree $d=O(\sqrt{s}\cdot\log(1/\varepsilon))$. That is, there exists a degree-$d$ polynomial $p\colon\boolpm^n\to\real$ over the reals such that for every $z\in\boolpm^n$,
		\[
			\left| p(z)-F(z)\right| \leq \varepsilon.
		\]
	\end{theorem}
	
	Note that \Cref{thm:approx-polynomial} still holds if we use $\bool$ as the Boolean basis.
    
	\subsection{Communication complexity}
	
	We use standard definitions from communication complexity. In this paper we consider the standard two party model of Yao and its generalizations to multiparty setting. We denote deterministic communication complexity of a Boolean function by $D(f)$ in the two party setting. We refer to \citep{Kushilevitz-Nisan97} for 
	standard definitions from communication complexity.
	
	    
	
	
	\begin{definition}
	    Let $f\colon\bool^{n}\to\bool$ be a Boolean function. The \emph{communication matrix of $f$}, namely $M_f$, is a $2^{n/2}\times2^{n/2}$ matrix defined by $\prn{M_f}_{x,y}:=f(x,y)$.
	\end{definition}
	
	\begin{definition}
	    A \emph{rectangle} is a set of the form $A\times B$, for $A,B\subseteq\bool^n$. A \emph{monochromatic rectangle} is a rectangle $S$ such that for all pairs $(x,y)\in S$ the value $f(x,y)$ is the same.
    \end{definition}
	
	\begin{lemma}
	    Let $\Pi$ be a protocol that computes $f\colon\bool^{n}\to\bool$ with at most $D$ bits of communication. Then, $\Pi$ induces a partition of $M_f$ into at most $2^D$ monochromatic rectangles.
	\end{lemma}
	
	Given a protocol, its \emph{transcript} is the sequence of bits communicated.
	
	\begin{lemma}
	    For every transcript $z$ of some communication protocol, the set of inputs $(x,y)$ that generate $z$ is a rectangle.
	\end{lemma}
	
	

    Below, we recount the definitions of two multiparty communication models used in this work, namely the number-on-forehead and the number-in-hand models.
	
	\begin{definition}[``Number-on-forehead'' communication model; informal]
	    In the $k$-party ``number-on-forehead'' communication model, 
	    there are $k$ players and $k$ strings $x_1,\dots,x_k\in\bool^{n/k}$ and player $i$ gets all the strings except for $x_i$. The players are interested in computing a value $f\cprn{x_1,\dots,x_k}$, where $f\colon\bool^{n}\to\bool$ is some fixed function. 
	    We denote by $D^{(k)}(f)$ the number of bits that must be exchanged by the best possible number on forehead protocol solving $f$.
    \end{definition}
        
    
    We also use the following weaker communication model.
    
    \begin{definition}[``Number-in-hand'' communication model; informal]
	    In the $k$-party ``number-in-hand'' communication model,
	    there are $k$ players and $k$ strings $x_1,\dots,x_k\in\bool^{n/k}$ and player $i$ gets only $x_i$.
	    The players are interested in computing a value $f\cprn{x_1,\dots,x_k}$, where $f\colon\bool^{n}\to\bool$ is some fixed function.
	    We denote by $D^{(k\text{-}\mathsf{NIH})}(f)$ the number of bits that must be exchanged by the best possible communication protocol.
    \end{definition}
    
    Note that $D^{(k\text{-}\mathsf{NIH})}(f)\leq\prn{1-1/k}\cdot n+1$, for any $n$-variate Boolean function $f$, as if $k-1$ players write on the blackboard their string, then the player that did not reveal her input may compute $f\cprn{x_1,\dots,x_k}$ on her own and then publish it.
    
    For the communication models mentioned above, there are also bounded-error randomized versions, denoted by $R_\delta$, $R_\delta^{(k)}$, and $R^{(k\text{-}\mathsf{NIH})}_\delta$, respectively, where $0<\delta<1$ is an upper bound on the error probability of the protocol. In this setting, the players have access to some shared random string, say $r$, and the aforementioned error probability of the protocol is considered over the possible choices of $r$. Moreover, we require the error to be at most $\delta$ on each fixed choice of inputs.
    
    We can extend the definitions of the communication complexity measures, defined above, to classes of Boolean functions, in a natural way. That is, for any communication complexity measure $M\in\brkts{D,D^{(k)},D^{(k\text{-}\mathsf{NIH})},R_\delta,R_\delta^{(k)},R_\delta^{(k\text{-}\mathsf{NIH})}}$ and for any class of Boolean functions $\mathcal{G}$, we may define
    \[
        M\cprn{\mathcal{G}}:=\max_{g\in\mathcal{G}}M\cprn{g}.
    \]
	We note that throughout this paper, we denote by $n$ the number of input bits for the function regardless the communication models. In the $k$-party communication setting (either NOF or NIH), we assume without loss of generality that $n$ is divisible by $k$.
	
	\subsection{Pseudorandomness}
	
	A PRG against a class of functions $\mathfrak{C}$ is a deterministic procedure $G$ mapping short Boolean strings (seeds) to longer Boolean strings, so that $G$'s output ``looks random" to every function in $\mathfrak{C}$.
	
    \begin{definition}[Pseudorandom generators]
    	Let $G\colon\boolpm^{\ell}\to\boolpm^n$ be a function, $\mathfrak{C}$ be a class of Boolean functions, and $0<\varepsilon<1$. We say that $G$ is a \emph{pseudorandom generator of seed length $\ell$ that $\varepsilon$-fools $\mathfrak{C}$} if, for every function $f\in\mathfrak{C}$, it is the case that
    	\[
    	    \abs{\Exp_{z\sim\boolpm^\ell}\csqbra{f\cprn{G\cprn{z}}}-
    	    \Exp_{x\sim\boolpm^n}\csqbra{f\cprn{x}}}\leq\varepsilon.
    	\]
    \end{definition}
	
	A PRG $G$ outputting $n$ bits is called \emph{explicit} if $G$ can be computed in $\poly(n)$ time. All PRGs stated in this paper are explicit.
	
	\subsection{Learning}
	
    For a function $f:\bool^n\to\bool$ and a distribution $\mathcal{D}$ supported over $\bool^n$, we denote by ${\rm EX}\cprn{f,\mathcal{D}}$ a randomized oracle that outputs independent identically distributed labelled examples of the form $\prn{x,f(x)}$, where $x\sim\mathcal{D}$.

    \begin{definition}[PAC learning model \cite{Val84}]\label{def:pac}
        Let $\mathcal{C}$ be a class of Boolean functions. We say that a randomized algorithm $A$ \emph{learns} $\mathcal{C}$ if, when $A$ is given oracle access to ${\rm EX}\cprn{f,\mathcal{D}}$ and inputs $1^n$, $\varepsilon$, and $\delta$, the following holds. For every $n$-variate function $f \in \mathcal{C}$, distribution $\mathcal{D}$ supported over $\bool^n$, and real-valued parameters $\varepsilon > 0$ and $\delta > 0$, $A^{{\rm EX}\prn{f,\mathcal{D}}}(1^n, \varepsilon, \delta)$ outputs  with probability at least $1-\delta$ over its internal randomness and the randomness of the example oracle ${\rm EX}\cprn{f,\mathcal{D}}$ a description of a hypothesis $h:\brkts{0,1}^n\to\brkts{0,1}$ such that
        \begin{align*}
            \Prob_{x\sim\mathcal{D}}\csqbra{f\cprn{x}=h\cprn{x}} \;\geq\; 1-\varepsilon.
        \end{align*}
        The \emph{sample complexity} of a learning algorithm is the maximum number of random examples from ${\rm EX}\cprn{f,\mathcal{D}}$ requested during its execution.
    \end{definition}
	


	
	\section{Lower bounds}\label{sec:lower-bounds}
	
	In this section, we prove an average-case lower bound for the generalized inner product function against $\FORMULA\circ\mathcal{G}$, where $\mathcal{G}$ is the set of functions that have low-cost randomized communication protocols in the number-on-forehead setting. This corresponds to Item 1 of Theorem \ref{thm:main_lbs}. Items 2 and 3 rely on our PRG constructions, and the proofs are deferred to Section \ref{sec:prgs}.
	
	\begin{theorem}\label{thm:lb-main}
		For any integer $k\geq 2$, $s>0$ and any class of functions $\mathcal{G}$, let $C\colon\boolpm^{n}\to\boolpm$ be a function in $\FORMULA[s]\circ\mathcal{G}$ such that
		\[
		    \Prob_{x\sim\boolpm^{n}}\left[C(x)=\GIP_{n}^{k}(x)\right]\geq 1/2+\varepsilon.
		\]
		Then
		\[
	  		s = \Omega\left( \frac{n^2}{k^2\cdot16^{k}\cdot\left(R^{(k)}_{\varepsilon/(2n^2)}(\mathcal{G})+\log n\right)^2\cdot\log^2(1/\varepsilon)}\right).
		\]
	\end{theorem}
	
	We need a couple useful lemmas from~\cite{Tal16}, whose proofs are presented in \Cref{subsec:formula-lemmas} (\Cref{lem:lb-Tal-1-app} and \Cref{lem:lb-Tal-2-app}) for completeness.
	
	\begin{lemma}[\cite{Tal16}]\label{lem:lb-Tal-1}
	    Let $\mathcal{D}$ be a distribution over $\boolpm^n$, and let $f,C\colon\boolpm^n\to\boolpm$ be such that
	    \[
	        \Prob_{x\sim\mathcal{D}}[C(x)=f(x)]\geq 1/2+\varepsilon.
	    \] 
	    Let $\tilde{C}\colon\boolpm^n\to\real$ be a $\varepsilon$-approximating function of $C$, i.e., for every $x\in\boolpm^n$, $|C(x)-\tilde{C}(x)|\leq \varepsilon$. Then,
	    \[
	        \Exp_{x\sim\mathcal{D}}[\tilde{C}(x)\cdot f(x)]\geq \varepsilon.
	    \]
	\end{lemma}
	
	\begin{lemma}[\cite{Tal16}]\label{lem:lb-Tal-2}
	    Let $\mathcal{D}$ be a distribution over $\boolpm^n$ and let $\mathcal{G}$ be a class of functions. For $f\colon\boolpm^n\to\boolpm$, suppose that $D\colon\boolpm^n\to\boolpm\in \FORMULA[s]\circ\mathcal{G}$ is such that 
	    \[
	        \Prob_{x\sim\mathcal{D}}[D(x)=f(x)]\geq 1/2+\varepsilon_0.
	    \]
	    Then there exists some $h\colon\boolpm^n\to\boolpm\in \XOR_{O\left(\sqrt{s}\cdot\log(1/\varepsilon_0)\right)}\circ \mathcal{G}$ such that
	    \[
	        \Exp_{x\sim\mathcal{D}}[h(x)\cdot f(x)]\geq \frac{1}{s^{O\left(\sqrt{s}\cdot\log(1/\varepsilon_0)\right)}}.
	    \]
	\end{lemma}
    
    We also need the following communication-complexity lower bound for $\GIP$.
    
    \begin{theorem}[{\cite[Theorem 2]{BNS92}}]\label{thm:lb-GIP-bound}
        For any $k\geq 2$, any function that computes $\GIP_{n}^{k}$ on more than $1/2+\delta$ fraction of the inputs (over uniformly random inputs) must have $k$-party deterministic communication complexity at least $\Omega \! \left(n/(k \cdot 4^k) - \log(1/\delta)\right)$.
    \end{theorem}
    
    We first show that if a function correlates with some small formula, whose leaves are functions with low \emph{randomized} communication complexity, then it also correlates non-trivially with some function of relatively low \emph{deterministic} communication complexity.
    
    \begin{lemma}\label{thm:lb-proof-1}
        For any distribution $\mathcal{D}$ over $\boolpm^n$, and any class of functions $\mathcal{G}$, let $f\colon\boolpm^n\to\boolpm$ and $C\colon\boolpm^n\to\boolpm\in \FORMULA[s]\circ\mathcal{G}$ be such that 
	    \[
	        \Prob_{x\sim\mathcal{D}}[C(x)=f(x)]\geq 1/2+\varepsilon.
	    \]
	    Then there exists a function $h$, with $k$-party deterministic communication complexity at most
	    \[
	        O\left(R^{(k)}_{\varepsilon/(2s)}(\mathcal{G})\cdot\sqrt{s}\cdot\log(1/\varepsilon)\right),
	    \]
	    such that
	    \[
	        \Prob_{x\sim\mathcal{D}}[h(x)=f(x)]\geq 1/2+1/s^{O(\sqrt{s}\cdot\log(1/\varepsilon))}.
	    \]
    \end{lemma}
    
    \begin{proof}
        Let $C=F(g_1,g_2\dots,g_s)$ be the function in $\FORMULA[s]\circ\mathcal{G}$, where $F$ is a formula and $g_1,g_2,\dots,g_s$ are leaf functions from the class $\mathcal{G}$. For each  $g_i$, consider a $k$-party randomized protocol $ \Pi_i$ of cost at most $R=R^{(k)}_{\varepsilon/(2s)}(\mathcal{G})$ that has an error $\varepsilon/(2s)$.
        Now consider the following function
        \[
            \tilde{C}(x) \vcentcolon=\Exp_{\Pi_1,\Pi_2,\dots,\Pi_s}\left[D(x)\right],
        \]
        where
        \[
            D(x) \vcentcolon= F(\Pi_1(x),\Pi_2(x),\dots,\Pi_s(x)).
        \]
        Note that for any fixed choice of $(\Pi_1,\Pi_2,\dots,\Pi_s)$, $D$ is a formula whose leaves are functions with \emph{deterministic} communication complexity at most $R$. Next, we show the following.
        
        \begin{claim}\label{claim:lb-approximating}
            The function $\tilde{C}$ $\varepsilon$-approximates $C$.
        \end{claim}
        
        \begin{proof}[Proof of \Cref{claim:lb-approximating}]
            First note that since each $\Pi_i$ is a $(\varepsilon/(2s))$-error randomized protocol, by taking the union bound over the $s$ leaf functions, we have that for every input $x\in\boolpm^n$,
            \[
                \Prob_{\Pi_1,\Pi_2,\dots,\Pi_s}[\Pi_1(x)=g_1(x) \land \Pi_2(x)=g_2(x) \land\dots\land\Pi_s(x)=g_s(x)]\geq 1-\varepsilon/2.
            \]
            Denote by $\mathcal{E}$ the event $\Pi_1(x)=g_1(x) \land \Pi_2(x)=g_2(x)\land\dots\land\Pi_s(x)=g_s(x)$. We have for every $x\in\boolpm^n$,
            \begin{align*}
                \tilde{C}(x) &= \Exp_{\Pi_1,\Pi_2,\dots,\Pi_s}\left[D(x)\right]\\
                &=\Exp\left[D(x)\given \mathcal{E}\right]\cdot\Prob[\mathcal{E}]+\Exp\left[D(x)\given \neg\mathcal{E}\right]\cdot\Prob[\neg\mathcal{E}]\\
                &=C(x)\cdot\Prob[\mathcal{E}]+\Exp\left[D(x)\given \neg\mathcal{E}\right]\cdot\Prob[\neg\mathcal{E}].
            \end{align*}
            On the one hand, we have
            \[
                \tilde{C}(x)=C(x)\cdot\Prob[\mathcal{E}]+\Exp\left[D(x)\given \neg\mathcal{E}\right]\cdot\Prob[\neg\mathcal{E}]\leq  C(x) + \varepsilon/2.
            \]
            On the other hand, we get
            \[
                \tilde{C}(x)=C(x)\cdot\Prob[\mathcal{E}]+\Exp\left[D(x)\given \neg\mathcal{E}\right]\cdot\Prob[\neg\mathcal{E}]\geq C(x)\cdot(1-\varepsilon/2)+(-1)\cdot(\varepsilon/2)\geq C(x)-\varepsilon.
            \]
            This completes the proof of the claim.
        \end{proof}
        
        Now by \Cref{claim:lb-approximating} and \Cref{lem:lb-Tal-1}, we have
        \begin{equation}\label{eq:lb-eq-1}
	        \Exp_{x\sim\mathcal{D}}[\tilde{C}(x)\cdot f(x)]\geq \varepsilon.
        \end{equation}
        By the definition of $\tilde{C}$, \Cref{eq:lb-eq-1} implies that there exists some $D$, which is a formula whose leaves are functions with \emph{deterministic} communication complexity at most $R$, such that
        \[
            \Exp_{x\sim\mathcal{D}}[D(x)\cdot f(x)]\geq \varepsilon,
        \]
        which implies
        \[
            \Prob_{x\sim\mathcal{D}}[D(x)=f(x)]\geq 1/2+ \varepsilon/2.
        \]
        Then by \Cref{lem:lb-Tal-2}, there exists a function $h$, which can be expressed as the $\XOR$ of at most $O(\sqrt{s}\cdot\log(1/\varepsilon))$ leaf functions in $D$, such that
        \[
            \Exp_{x\sim\mathcal{D}}[h(x)\cdot f(x)]\geq \frac{1}{s^{O(\sqrt{s}\cdot\log(1/\varepsilon))}},
        \]
        which again implies
        \[
            \Prob_{x\sim\mathcal{D}}[h(x)=f(x)]\geq \frac{1}{2} + \frac{1}{s^{O\left(\sqrt{s}\cdot\log(1/\varepsilon)\right)}}.
        \]
        Finally, note that the $k$-party deterministic communication complexity of $h$ is at most 
        \[
            O(R\cdot\sqrt{s}\cdot\log(1/\varepsilon)),
        \]
        where $R=R^{(k)}_{\varepsilon/(2s)}(\mathcal{G})$.
    \end{proof}
    
    We are now ready to show \Cref{thm:lb-main}.
    
    \begin{proof}[Proof of \Cref{thm:lb-main}]
        Consider \Cref{thm:lb-proof-1} with $f$ being $\GIP_{n}^{k}$ and $\mathcal{D}$ being the uniform distribution. Consider \Cref{thm:lb-GIP-bound} with $\delta=1/s^{O(\sqrt{s}\cdot\log(1/\varepsilon))}$. We have
        \[
            O\left(R^{(k)}_{\varepsilon/(2s)}(\mathcal{G})\cdot\sqrt{s}\cdot\log(1/\varepsilon)\right)\geq n/(k4^k) - O\left(\sqrt{s}\cdot\log(s)\cdot\log(1/\varepsilon))\right),
        \]
        which implies
        \[
            s\geq \Omega\left( \frac{n^2}{k^2\cdot 16^{k}\cdot\left(R^{(k)}_{\varepsilon/(2n^2)}(\mathcal{G})+\log n\right)^2\cdot\log^2(1/\varepsilon)}\right).\qedhere
        \]
    \end{proof}
	
	\section{Pseudorandom generators}\label{sec:prgs}
	
	Some of our PRGs are obtained from a general framework that allows us to reduce the task of fooling $\FORMULA\circ\mathcal{G}$ to the task of fooling the class of functions which are the parity or conjunction of few functions from $\mathcal{G}$.
	
	\subsection{The general framework}
	
	We show that in order to get a PRG for the class of subquadratic-size formulas with leaf gates in $\mathcal{G}$, it suffices to get a PRG for very simple sublinear-size formulas: either $\XOR\circ\mathcal{G}$ or $\AND\circ\mathcal{G}$.  
	
	\begin{theorem}[PRG for $\FORMULA\circ\mathcal{G}$ from PRG for $\XOR\circ\mathcal{G}$ or $\AND\circ\mathcal{G}$]\label{thm:prg-main}
		Let $\mathcal{G}$ be a class of gates on $n$ bits. For any integer $s>0$ and any $0<\varepsilon<1$, there exists a constant $c>0$ such that the following holds. If a distribution $\mathcal{D}$ over $\boolpm^n$ $\left(2^{-c\cdot\sqrt{s}\cdot\log (s)\cdot \log(1/\varepsilon)}\right)$-fools the $\XOR$ (parity) or the $\AND$ (conjunction) of $c\cdot\sqrt{s}\cdot \log(1/\varepsilon)$ arbitrary functions from $\mathcal{G}$, then $\mathcal{D}$ also $\varepsilon$-fools $\FORMULA[s]\circ\mathcal{G}$.
	\end{theorem}
	
	\begin{proof}\renewcommand{\qedsymbol}{$\square$ (\Cref{thm:prg-main})}
    	We first show the case where $\mathcal{D}$ fools the parity of a few functions from $\mathcal{G}$. The proof can be easily adapted to the case of conjunction.
    	
    	Let $C=F(g_1,g_2\dots,g_s)$ be a function in $\FORMULA[s]\circ\mathcal{G}$, where $F$ is a formula, and $g_1,g_2,\dots,g_s$ are functions from the class $\mathcal{G}$. Let $\uniform$ be the uniform distribution over $\boolpm^n$. 
    	We need to show
    	\begin{equation}\label{eq:prg-1}
    		\Exp[C(\mathcal{D})] \myapprox{\varepsilon} \Exp[C(\uniform)].
    	\end{equation}
    	Let $p$ be a $(\varepsilon/3)$-approximating polynomial for $F$ given by \Cref{thm:approx-polynomial}. Note that the degree of $p$ is
    	\[
    		d=O(\sqrt{s}\cdot\log(1/\varepsilon)).
    	\] 
    	Let us replace $F$, the formula part of $C$, with $p$ and let
    	\[
    		\tilde{C}\vcentcolon=p(g_1,g_2\dots,g_s).
    	\]
    	Since $\tilde{C}$ point-wisely approximates $C$, we have
    	\[
    		\Exp[\tilde{C}(\uniform)] \myapprox{\varepsilon/3} \Exp[C(\uniform)],
    	\]
    	and
    	\[
    		\Exp[\tilde{C}(\mathcal{D})] \myapprox{\varepsilon/3} \Exp[C(\mathcal{D})].
    	\]
    	Then to show \Cref{eq:prg-1}, it suffices to show
    	\[
    		\Exp[\tilde{C}(\mathcal{D})] \myapprox{\varepsilon/3} \Exp[\tilde{C}(\uniform)].
    	\]
    	We have
    	\begin{align*}
        	\label{eq:prg-2}
        	\Exp_{x\sim \mathcal{D}}[\tilde{C}(x)]&=\Exp_{x\sim D}\left[\sum_{\substack{S\subseteq[s]: \\ |S|\leq d}}\hat{p}(S)\cdot\prod_{i\in S}g_i(x)\right]\\
        	&=\sum_{\substack{S\subseteq[s]: \\ |S|\leq d}}\hat{p}(S)\cdot\Exp_{x\sim \mathcal{D}}\left[\prod_{i\in S}g_i(x)\right].
        	\numberthis
    	\end{align*}
    	Now note that for each $S\subseteq[s]$, $\prod_{i\in S}g_i(x)$ computes the $\XOR$ of at most $d$ functions from $\mathcal{G}$.
    	Using the fact the distribution $\mathcal{D}$ $\left(\delta=1/2^{c\cdot\sqrt{s}\cdot\log (s)\cdot \log(1/\varepsilon)}\right)$-fools the $\XOR$ of any $d$ functions from $\mathcal{G}$, we get
    	\begin{align*}
        	\Exp_{x\sim \mathcal{D}}[\tilde{C}(x)]&=\sum_{\substack{S\subseteq[s]: \\ |S|\leq d}}\hat{p}(S)\cdot\Exp_{x\sim D}\left[\prod_{i\in S}g_i(x)\right]\\
        	&=\sum_{\substack{S\subseteq[s]: \\ |S|\leq d}}\hat{p}(S)\cdot\left(\Exp_{x\sim \uniform}\left[\prod_{i\in S}g_i(x)\right]+\delta_S\right) \tag{where $|\delta_S|\leq \delta$} \\
        	&=\sum_{\substack{S\subseteq[s]: \\ |S|\leq d}}\left(\hat{p}(S)\cdot\Exp_{x\sim \uniform}\left[\prod_{i\in S}g_i(x)\right]+\hat{p}(S)\cdot\delta_S\right)\\
        	&=\sum_{\substack{S\subseteq[s]: \\ |S|\leq d}}\hat{p}(S)\cdot\Exp_{x\sim \uniform}\left[\prod_{i\in S}g_i(x)\right]+\sum_{\substack{S\subseteq[s]: \\ |S|\leq d}} \hat{p}(S)\cdot\delta_S\\
        	&=\Exp_{x\sim \uniform}[\tilde{C}(x)]+\sum_{\substack{S\subseteq[s]: \\ |S|\leq d}} \hat{p}(S)\cdot \delta_S.
    	\end{align*}
    	It remains to show
    	\[
    		\left| \sum_{\substack{S\subseteq[s]: \\ |S|\leq d}} \hat{p}(S)\cdot \delta_S \right| \leq \varepsilon/3.
    	\]
    	Note that because $p(z)\in[1-\varepsilon/3,1+\varepsilon/3]$ for every $z\in\boolpm^s$, we have
    	\[
    		|\hat{p}(S)|=\left|\Exp_{z\sim \boolpm^s}\left[p(z)\cdot\prod_{i\in S} z_i\right]\right| \leq 1+\varepsilon/3 < 2.
    	\]
    	Then,
    	\[
    		\left|\sum_{\substack{S\subseteq[s]: \\ |S|\leq d}} \hat{p}(S)\cdot \delta_S \right| \leq  \sum_{\substack{S\subseteq[s]: \\ |S|\leq d}} |\hat{p}(S)|\cdot |\delta_S|\leq \delta \cdot \sum_{\substack{S\subseteq[s]: \\ |S|\leq d}} |\hat{p}(S)| \leq \delta\cdot s^{O(\sqrt{s}\cdot\log(1/\varepsilon))} \leq \varepsilon/3,
    	\]
    	where the last inequality holds for some sufficiently large constant $c$.
    	
    	To show the case of conjunction, we can write the approximating polynomial as the sum of all degree-$d$ monomials, each of which is the $\AND$ of at most $d$ variables. One way to do this is to use the domain $\bool$ instead of $\boolpm$ in the above argument. We need to show that the coefficients in this case still have small magnitude.
    	
    	\begin{claim}\label{claim:prg-magnitude}
    		Let $p\colon\boolpm^{n}\to\real$ be a degree-$d$ polynomial of the form
    		\[
    			p(x)=\sum_{\substack{S\subseteq[n]: \\ |S|\leq d}} \hat{p}(S)\cdot \prod_{i\in S} x_i,
    		\]
    		and let $q\colon\bool^{n}\to\real$ be the corresponding polynomial of $p$ over the domain $\{0,1\}^n$, of the form
    		\[
    			q(y)=\sum_{\substack{T\subseteq[n]: \\ |T|\leq d}} \hat{q}(T)\cdot \prod_{i\in T} y_i.
    		\]
    		Then,
    		\[
    			|q|_1 = \sum_{\substack{T\subseteq[n]: \\ |T|\leq d}} |\hat{q}(T)| \leq n^{O(d)}\cdot \max_{\substack{S\subseteq[n]: \\ |S|\leq d}} |\hat{p}(S)|.
    		\]
    	\end{claim}
    	
    	\begin{proof}\renewcommand{\qedsymbol}{$\square$ (\Cref{claim:prg-magnitude})}
    		We have
    		\begin{align*}
    			q(y_1,y_2,\dots,y_n) &= p(1-2y_1,1-2y_2,\dots,1-2y_n)\\
    			&= \sum_{\substack{S\subseteq[n]: \\ |S|\leq d}} \hat{p}(S)\cdot \prod_{i\in S} (1-2y_i)\\
    			&= \sum_{\substack{S\subseteq[n]: \\ |S|\leq d}} \hat{p}(S)\cdot \left(\sum_{\ell\in \{0,1\}^{|S|}}\prod_{\substack{j\in S:\\ \ell_j=1}}-2y_j \right)\\
    			&= \sum_{\substack{S\subseteq[n]: \\ |S|\leq d}} \sum_{\ell\in \{0,1\}^{|S|}} \hat{p}(S)\cdot(-2)^{|\ell|}\cdot \prod_{\substack{j\in S:\\ \ell_j=1}} y_j. \tag{where $|\ell| = \sum_{i=1}^{|S|}\ell_i$}
    		\end{align*}
    		For a pair $(S,\ell)$ where $S \subseteq[n]$, $|S|\leq d$ and $\ell\in\{0,1\}^{|S|}$, let us define the polynomial $q_{(S,\ell)}$ as
    		\[
    			q_{(S,\ell)}(y)= \hat{p}(S)\cdot(-2)^{|\ell|}\cdot \prod_{\substack{j\in S:\\ \ell_j=1}} y_j.
    		\]
    		Note that there are at most $n^{d}\cdot2^{d}$ many pairs of such $(S,\ell)$'s and for each $(S,\ell)$, we have
    		\begin{align*}
    			|q_{(S,\ell)}|_1 = \left|\hat{p}(S)\cdot(-2)^{|\ell|}\right|\leq 2^{d}\cdot |\hat{p}(S)|.
    		\end{align*}
    	    Finally we have
    		\[
    			|q|_1 = \left|\sum_{(S,\ell)} q_{(S,\ell)}\right|_1\leq \sum_{(S,\ell)} |q_{(S,\ell)}|_1\leq n^{d}\cdot 2^d \cdot 2^d \cdot \max_{\substack{S\subseteq[n]: \\ |S|\leq d}} |\hat{p}(S)|,
    		\]
    		as desired.
    	\end{proof}
    	
    	This completes the proof of \Cref{thm:prg-main}.
	\end{proof}

    \subsection{Formulas of low-communication functions in the number-in-hand setting}
    
    In this subsection, we will use $\bool$ as the Boolean basis.
    
	\begin{theorem}\label{thm:prg-formula-lc}
		For any integers $k\geq 2$, $s>0$ and any $0<\varepsilon<1$, let $\mathcal{G}$ be the class of functions that have $k$-party number-in-hand $(\varepsilon/6s)$-error randomized communication protocols of cost at most $R$. There exists a PRG that $\varepsilon$-fools $\FORMULA[s]\circ \mathcal{G}$ with seed length
		\[
			n/k+O\left(\sqrt{s}\cdot (R +\log (s))\cdot \log(1/\varepsilon) +\log(k)\right)\cdot\log(k).
		\]
	\end{theorem}
    
    We need the following PRG that fools single functions with low communication complexity in the number-in-hand model. The proof is presented in \Cref{appendix:prg_cc} (\Cref{lem:prg-lc-app}) for completeness.

	\begin{theorem}[\cite{ASWZ96,INW94}]\label{thm:prg-lc}
		For any $k\geq 2$, there exists a PRG that $\delta$-fools any $n$-bits functions with $k$-party number-in-hand deterministic communication complexity of at most $D'$, with seed length
		\[
			n/k+O\left(D' + \log(1/\delta)+\log(k)\right)\cdot\log(k).
		\]
	\end{theorem}

    Next, we show a PRG for $\FORMULA\circ\mathcal{G}$, where $\mathcal{G}$ is the class of functions with low-cost communication protocols in the number-in-hand setting. We first show for the case of deterministic protocols.
    
    \begin{theorem}\label{thm:prg-formula-lc-det}
		For any integers $k\geq 2$ and $s>0$, let $\mathcal{G}$ be the class of functions whose $k$-party number-in-hand deterministic communication complexity are at most $D$. There is a PRG that $\varepsilon$-fools $\FORMULA[s]\circ \mathcal{G}$ of size $s$ with seed length
		\[
			n/k+O\left(\sqrt{s}\cdot \log(1/\varepsilon)\cdot (D +\log (s)) +\log(k)\right)\cdot\log(k).
		\]
	\end{theorem}
	
	\begin{proof}
		By \Cref{thm:prg-main}, it suffices to show a PRG that $\left(\delta=1/2^{c\cdot\sqrt{s}\cdot\log (s)\cdot \log(1/\varepsilon)}\right)$-fools every function that is the $\XOR$ of $t=c\cdot\sqrt{s}\cdot \log(1/\varepsilon)$ arbitrary functions from $\mathcal{G}$. Note that such a function has deterministic communication complexity at most $D'=t\cdot D$. Then \Cref{thm:prg-formula-lc-det} follows from \Cref{thm:prg-lc}.
	\end{proof}
	
	We now establish the randomized case.
    
    \begin{proof}[Proof of \Cref{thm:prg-formula-lc}]
    	Let $C$ be a function in $\FORMULA[s]\circ \mathcal{G}$. For each of the leaf functions in $C$, consider a $k$-party number-in-hand randomized protocol of cost at most $R$ that has an error at most $\varepsilon/(6s)$. By taking a union bound over the $s$ leaf functions and by viewing a randomized protocol as a distribution of deterministic protocols (as shown in the proof of \Cref{claim:lb-approximating}), we get the following which is a (point-wisely) $(\varepsilon/3)$-approximating function for $C$:
        \[
            \tilde{C}(x) \vcentcolon=\sum_{i} p_i \cdot D_i(x),
        \]
        where each $p_i\in[0,1]$ is some probability density value (so $\sum_i p_i =1$), and each $D_i$ is a formula whose leaves are functions with \emph{deterministic} communication complexity at most $R$. Then to $\varepsilon$-fool $C$, it suffices to $(\varepsilon/3)$-fool its $(\varepsilon/3)$-approximating function $\tilde{C}$. Also, since $\tilde{C}$ is a convex combination of the $D_i$'s, it suffices to $(\varepsilon/3)$-fools all the $D_i$'s. We will do this using the PRG form \Cref{thm:prg-formula-lc-det}. We get that there exists a PRG that $(\varepsilon/3)$-fools each $D_i$ with seed length
        \[
            n/k+O\left(\sqrt{s}\cdot(R+\log(s))\cdot\log(1/\varepsilon)+\log(k)\right)\cdot\log(k),
        \]
        as desired.
    \end{proof}

    \subsection{Applications: Fooling formulas of SYMs, LTFs, XORs, and AC\texorpdfstring{$^0$ circuits}{AC0}}

    \subsubsection{\texorpdfstring{$\FORMULA\circ\SYM$}{FORMULA-SYM} and \texorpdfstring{$\FORMULA\circ\LTF$}{FORMULA-LTF}}
	
	Here, we show how the PRG in \Cref{thm:prg-formula-lc} implies PRGs for $\FORMULA\circ\LTF$ and $\FORMULA\circ\SYM$.
	
	\begin{theorem}\label{thm:prg-formula-sym}
		For any size $s>0$ and $0<\varepsilon<1$, there exists a PRG that $\varepsilon$-fools $\FORMULA[s]\circ \LTF$ with seed length
        \[
			O\left(n^{1/2}\cdot s^{1/4}\cdot \log(n)\cdot \log(n/\varepsilon)\right).
		\]
		For $\FORMULA[s]\circ \SYM$, the seed length is
		\[
		    O\left(n^{1/2}\cdot s^{1/4}\cdot \log(n)\cdot \log(1/\varepsilon) \right).
		\]
	\end{theorem}
	
	We need the fact that the class of $\LTF$ has low communication complexity in the number-in-hand model. Consider the following $k$-party $\SG_m$ problem where the $i$-th party holds a $m$-bit number $z_i$ in hand and they want to determine whether $\sum_{i=1}^k z_i > \theta$, where $\theta$ is a fixed number known to all the parties. Nisan~\cite{Nis94} gave an efficient randomized protocol (with public randomness) for this problem.
    
    
    \begin{theorem}[\cite{Nis94}\footnote{Viola~\cite{Vio15} gave a $\delta$-error randomized protocol for the $k$-party $\SG_m$ problem of cost $O(k\cdot\log(k)\cdot\log(m/\delta))$, which is better than Nisan's protocol when $k=m^{o(1)}$.}]\label{thm:prg-k-party-sum}
	    Let $m>0$ be an integer. For any integer $2\leq k\leq m^{O(1)}$, and any $0<\delta<1$, there exists a $\delta$-error randomized protocol of cost $O(k\cdot\log(m)\cdot\log(m/\delta))$ for the $k$-party $\SG_m$ problem.
	\end{theorem}
	    
    By \Cref{thm:prg-k-party-sum} and the fact that every linear threshold function on $n$ bits has a representation such that the weights are $O(n\log(n))$ integers~\cite{MTT61}, we get the following.
    
    \begin{corollary}\label{cor:prg-k-party-LTF}
         For every $k\geq 2$ and $0<\delta<1$, the $k$-party number-in-hand $\delta$-error randomized communication complexity of $\LTF$ is $O(k\cdot\log(n)\cdot\log(n/\delta))$.
    \end{corollary}
    
	\begin{proof}[Proof of \Cref{thm:prg-formula-sym}]
	    By \Cref{cor:prg-k-party-LTF} and \Cref{thm:prg-formula-lc}, for \emph{every} $k\geq 2$ we get a PRG for $\FORMULA\circ\LTF$ of seed length
        \[
    		n/k+O\left(\sqrt{s}\cdot k\cdot \log(n) \cdot\log(ns/\varepsilon)\cdot\log(1/\varepsilon)+\log(k)\right)\cdot\log(k).
    	\]
	    By choosing 
	    \[
	        k=\frac{n^{1/2}}{s^{1/4}\cdot\log(n)\cdot\log(n/\varepsilon)},
	    \]
	    the claimed seed length follows from a simple calculation.
	
	    For $\FORMULA\circ\SYM$, note that every $n$-bit symmetric function has a deterministic $k$-party number-in-hand communication protocol of cost at most $k\cdot\log(n)$. Then the rest can be shown using a similar argument as above (by choosing $k=n^{1/2}/\left(s^{1/4}\cdot\log(n)\right)$).
	\end{proof}

       


    \subsubsection{\texorpdfstring{$\FORMULA\circ\XOR$}{FORMULA-XOR}}
    
    For the case of $\FORMULA\circ\XOR$, we get a PRG with better seed length. 
    
	\begin{theorem}\label{thm:prg-formula-xor}
		For any size $s>0$ and $0<\varepsilon<1$, there exists a PRG that $\varepsilon$-fools $\FORMULA[s]\circ \XOR$ with seed length
		\[
			O\left(\sqrt{s}\cdot\log(s)\cdot\log(1/\varepsilon)+\log(n)\right).
		\]
	\end{theorem}
	
	\begin{proof}
	    By \Cref{thm:prg-main}, to fool $\FORMULA[s]\circ\mathcal{G}$, it suffices to $\left(\delta=1/2^{O\left(\sqrt{s}\cdot\log (s)\cdot \log(1/\varepsilon)\right)}\right)$-fool the $\XOR$ of a few functions from $\mathcal{G}$, where $\mathcal{G}$ in this case is the set of all $\XOR$ functions. Note that the $\XOR$ of any set of $\XOR$ functions simply computes some $\XOR$ function. Therefore, we can use small-bias distribution, which fools every $\XOR$ function, to fool $\FORMULA[s]\circ\XOR$. Finally, note that there are known constructions for $\delta$-bias distributions that use $O(\log(n/\delta))$ random bits (see e.g.~\cite{AGHP92}).
	\end{proof}
	
	Using the ``locality'' of this PRG for $\FORMULA\circ\XOR$, we get a lower bound for $\MCSP$ against subquadratic-size formulas of XORs. 
	
	\begin{theorem}\label{thm:prg-formula-xor-mcsp}
	    For every integer $s>0$, if $\MCSP$ on $N$-bit can be computed by some function in $\FORMULA[s]\circ \XOR$, then $s=\tilde{\Omega}(N^2)$.
	\end{theorem}
	
	\begin{proof}[Proof sketch]
	    There is a standard construction of $\delta$-bias distribution that is local (see e.g.~{\cite[Construction 3]{AGHP92}} and~{\cite[Fact 7]{CKLM19}}) in the following sense: there exists a circuit of size at most $\tilde{O}(\log(n/\delta)\cdot \log(n))$ such that given a seed of length $O(\log(n/\delta))$ and a index $j\in[n]$, outputs the $j$-th bit of the distribution. Local PRGs imply $\MCSP$ lower bounds (see~\cite[Section 3]{CKLM19}). 
	\end{proof}
	
    \subsubsection{\texorpdfstring{$\FORMULA\circ\AC^0$}{FORMULA-AC0}}

	Another application of \Cref{thm:prg-main} is to take $\mathcal{G}$ to be the set all functions that can be computed by small constant-depth circuits ($\AC^0$). Note the state-of-the-art PRG against size-$M$ depth-$d$ $\AC^0$ has a seed length of $\log^{d+O(1)}(Mn)\cdot\log(1/\varepsilon)$~\cite{ST19}. Below, let $\AC^0_{d,M}$ denote the class of depth-$d$ circuits of size at most $M$.
	
	\begin{theorem}\label{thm:prg-formula-ac0}
		For any size $s,m>0$ and $0<\varepsilon<1$, there exists a PRG that $\varepsilon$-fools $\FORMULA\circ \AC^0_{d,M}$ of size $s$ with seed length
		\[
			\log^{d+O(1)}(Mn)\cdot\sqrt{s}\cdot\log(s)\cdot\log(1/\varepsilon).
		\]
	\end{theorem}
	
	Moreover, by inspecting the construction of PRG in~\cite{ST19}, it is not difficult to see that the PRG is also local; there exists a circuit of size at most $\lambda=\log^{d+O(1)}(Mn)\cdot\log(1/\varepsilon)$ such that given a seed of length $O\log^{d+O(1)}(Mn)\cdot\log(1/\varepsilon)$ and a index $j\in[n]$, outputs the $j$-th bit of the PRG. As a result, we get $\MCSP$ lower bounds from the this PRG.
	
	\begin{theorem}\label{thm:prg-formula-ac0-mcsp}
		 For every $s,d,M \in \mathbb{N}$, if $\MCSP$ on $N$-bit can be computed by some function in $\FORMULA[s]\circ \AC^0_{d,M}$, then
		 \[
		   s\geq N^2/\log^{2d+O(1)}(Mn).
		 \]
	\end{theorem}
	
	\subsection{Formulas of low number-on-forehead communication leaf gates}
	
	In this section, we show a PRG with mild seed length for formulas of functions with low \emph{multi-party number-on-forehead} communication complexity.
	
	\begin{theorem}\label{thm:mild-prg}
	    Let $\mathcal{G}$ be a class of $n$-bits functions. For any size $s>0$, there exists a PRG that $\varepsilon$-fools $\FORMULA[s]\circ \mathcal{G}$, with seed length
		\[
		    n-\frac{n}{O\left(\sqrt{s}\cdot k\cdot4^k\cdot \left( R^{(k)}_{\varepsilon/(2s)}(\mathcal{G})+\log(n)\right)\cdot\log(n/\varepsilon)\right)}.
		\]
	\end{theorem}
	
	The PRG is constructed using the hardness vs. randomness paradigm.
	    
	\subsubsection{Hardness based PRGs}
	
	We show how to construct the PRG using the average-case hardness result for formulas of functions with low multi-party communication complexity (\Cref{thm:lb-main}). We start with some notations. For $x\in\boolpm^m$ and an integer $k$ such that $k$ divides $m$, we consider a partition of $x$ into $k$ equal-sized consecutive blocks and write $x=x^{(1)},x^{(2)},\dots,x^{(k)}$, where $x^{(i)}\in\boolpm^{m/k}$ for each $i\in [k]$.
    	
	\begin{lemma}\label{lem:mild-prg}
	    For any integers $m,t,k>0$ such that $k$ divides $m,t$, let $\mathcal{G}$ be a class of functions on $mt+t$ bits, and let $G\colon\boolpm^{m\times t}\to\boolpm^{mt+t}$ be
		\begin{align*}
			&\quad G(x_1,x_2,\dots,x_t) \\
			&=\left(x_1^{(i)},x_2^{(i)},\dots,x_t^{(i)},\GIP_{m}^{k}\left(x_{(i-1)\cdot(t/k)+1}\right),\GIP_{m}^{k}\left(x_{(i-1)\cdot(t/k)+2}\right),\dots, \GIP_{m}^{k}\left(x_{i\cdot(t/k)+1}\right)\right)_{i\in[k]},
		\end{align*}
		where $x_1,x_2,\dots,x_t\in\boolpm^m$.
		Then $G$ is a PRG that $(t\cdot \varepsilon)$-fools $\FORMULA \circ \mathcal{G}$ of size 
		\[
			s=\Omega\left( \frac{m^2}{k^2\cdot16^{k}\cdot\left(R^{(k)}_{\varepsilon/(2m^2)}(\mathcal{G})+\log m\right)^2\cdot\log^2(1/\varepsilon)}\right).
		\]
	\end{lemma}
	
	\begin{proof}
		The high level idea is as follows. We argue that if there is a $\FORMULA \circ \mathcal{G}$ of the claimed size that breaks the PRG, then there is a $\FORMULA \circ \mathcal{G'}$ of the same size that computes $\GIP$ on $m$ bits, where $\mathcal{G'}$ has a $k$-party communication complexity that is at most that of $\mathcal{G}$ \emph{with respect to the $m$-bit input}, and hence contradicts the $\FORMULA \circ \mathcal{G'}$ complexity of the generalized inner product function. The resulting formula is obtained by fixing some input bits of the original $\FORMULA \circ \mathcal{G}$ which breaks the PRG.

		We use a hybrid argument. First consider the distribution given by $G$, where we replace each $\GIP(x_j)$ ($j\in[t]$) with a uniformly random bit; let us denote those random bits as $\uniform_{j}$ for $j\in[t]$ (note that this is just the uniform distribution). Then for each $j\in[t]$, define $H_j$ to be the distribution that we substitute back $\GIP(x_1),\GIP(x_2),\dots,\GIP(x_{j})$ for the corresponding uniform bits in the previous distribution. 

		For the sake of contradiction, suppose there exists a $\FORMULA \circ \mathcal{G}$ $C$ of size $s$ such that
		\[
			\left| \Prob[C(H_t)=1]-\Prob[F(H_0)=1] \right| > t\cdot\varepsilon.
		\]
		By the triangle inequality, there exists a $1\leq j\leq k$ such that
		\[
			\left| \Prob[C(H_j)=1]-\Prob[C(H_{j-1})=1] \right| >\varepsilon.
		\]
		Then by averaging, there exist some fixings of $x_1,\dots,x_{j-1},x_{j+1},\dots,x_{t}$ and $\uniform_{j+1},\dots,\uniform_{t}$ to $C$ such that the above inequality still holds. Let us denote by $C'$ the circuit obtained by $C$ after such fixings and assume without loss of generality $(k-1)t/k\leq j \leq t$. Then we have
		\begin{equation}\label{eq:distinguish}
		    \left| \Prob\left[C'\left(x_j^{(1)},x_j^{(2)},\dots,x_j^{(k)},\GIP(x_j)\right)=1\right]-\Prob\left[C'\left(x_j^{(1)},x_j^{(2)},\dots,x_j^{(k)},\uniform_j\right)=1\right] \right| > \varepsilon.
		\end{equation}
		By a standard ``unpredictability implies pseudorandomness" argument~\cite{Yao82}, we can show that there is some circuit $C''$, obtained from $C'$ by fixing some value for the last bit, that computes the generalized inner product function on $m$ bits with probability greater than $1/2+\varepsilon$ over uniformly random inputs. Note that the size of $C''$ is the same as $C'$ (hence also $C$) , and also $C''$ can be computed by some $\FORMULA \circ \mathcal{G'}$, where $R^{(k)}_{\delta}(\mathcal{G'})\leq R^{(k)}_{\delta}(\mathcal{G})$ for every $\delta$. This contradicts hardness of $\GIP$ for such circuits (\Cref{thm:lb-main}).
	\end{proof}
	
	We are now ready to prove \Cref{thm:mild-prg}.
	
	\begin{proof}[Proof of \Cref{thm:mild-prg}]
		Consider \Cref{lem:mild-prg}. Let $n=mt+t$, and we have $m=\left(\frac{n}{t}-1\right)$. Then 
		\Cref{lem:mild-prg} gives a PRG that $\varepsilon$-fools $\FORMULA\circ\mathcal{G}$ of size
		\begin{align*}
			s &= \Omega\left( \frac{m^2}{k^2\cdot16^{k}\cdot\left(R^{(k)}_{\varepsilon/(2m^2)}(\mathcal{G})+\log m\right)^2\cdot\log^2(t/\varepsilon)}\right)\\
			&\geq \Omega\left(\left(\frac{n}{t}\right)^2/\left(k^2\cdot16^{k}\cdot\left(R^{(k)}_{\varepsilon/(2n^2)}(\mathcal{G})+\log n\right)^2\cdot\log^2(n/\varepsilon)\right)\right),
		\end{align*}
		which yields
		\[
			t\geq \Omega\left(\frac{n}{\sqrt{s}\cdot k\cdot4^{k}\cdot\left(R^{(k)}_{\varepsilon/(2n^2)}(\mathcal{G})+\log n\right)\cdot\log(n/\varepsilon)}\right).
		\]
		Note that the seed length in this case is $n-t$.
	\end{proof}

	\subsubsection{\texorpdfstring{$\MKtP$}{MKtP} lower bounds}
	
	The PRG in \Cref{thm:mild-prg} is sufficient to give an $\MKtP$ lower bound for formulas of functions with low multi-party communication complexity.
	
	\begin{theorem}\label{thm:prg-mktp}
	    For any integer $s>0$ and any class of $N$-bit function $\mathcal{G}$, if $\MKtP$ on $N$-bit can be computed by some function $\FORMULA[s]\circ \mathcal{G}$, then
	    \[
            s= \frac{N^2}{k^2\cdot16^{k}\cdot R^{(k)}_{1/3}(\mathcal{G})\cdot \polylog(N)}.
	    \]
	\end{theorem}
	
	\begin{proof}
    	Let $C$ be a function in $\FORMULA\circ\mathcal{G}$ of size less than
    	\[
    	    \frac{N^2}{k^2\cdot16^{k}\cdot R^{(k)}_{1/3}(\mathcal{G})\cdot\log^c(N)}
    	\]
    	where $c>0$ is some sufficiently large constant. By \Cref{thm:mild-prg}, we have that there is a PRG that $(1/3)$-fools $C$ and its seed length is
    	\[
             N-\polylog(N).
    	\]
    	Also, since the PRG is polynomial-time computable, we get that for every seed, the output of the PRG has $\mathrm{Kt}$ complexity at most $\theta=N-\polylog(N)$. However, consider the $\MKtP$ function with a threshold parameter $\theta$; this function is not fooled by such a PRG, since it accepts every output of the PRG and rejects a uniformly random string with high probability. 
    \end{proof}
    
    \section{Satisfiability algorithms}\label{sec:sat}
	
	In this section, we will use $\bool$ as the Boolean basis.
	
	\subsection{Computational efficient communication protocols}
	
	\begin{definition}[Computational efficient communication protocols]
	    Let $t\colon\naturals\times\naturals\to\naturals$. We say that a two-party communication protocol is \emph{$t$-efficient} if for each of the parties, given an input $x$ and some previously sent messages $\pi\in\bool^{*}$, the next message to send can be computed in time $t(|x|,|\pi|)$ ($\bot$ is being output if there is no next message). We say that such a protocol is \emph{explicit} if $t(|x|,|\pi|)=2^{o(|x|+|\pi|)}$.
	\end{definition}

	\begin{lemma}\label{lem:sat-time-efficient}
        Let $f\colon\bool^n\to 1$ and let $\Pi$ be a $t$-efficient communication protocol for $f$ with communication cost at most $D$. Then the protocol tree of $\Pi$ can be output in time $O\left(D\cdot t\cprn{n/2,D}\cdot2^{n}\cdot2^D\right)$. That is, there exists an algorithm that outputs a list of all (partial and full) transcripts of length at most $D$ and the rectangles associated with each of the transcripts.
	\end{lemma}
	
	\begin{proof}
	    It suffices to show that, given an input $x\in\bool^{n/2}$ and a transcript $\ell\in\bool^{\leq D}$, we can decide whether $x$ belongs to the rectangle indexed by $\ell$ in time $D\cdot t\cprn{n/2,D}$.
	    Suppose $x$ is the input for Alice (resp. Bob), and we want to decide whether $x$ belongs to the rectangle indexed by $\pi$. We can carry out the communication task by simulating the behavior of Alice (resp. Bob) using the protocol $\Pi$ and simulating Bob's (resp. Alice's) behavior using the transcript $\pi$, and check whether the messages sent by Alice (resp. Bob) is consistent with the transcript $\pi$. This takes time at most $D\cdot t\cprn{n/2,D}$. To construct the tree, we do the above for every (partial and full) transcript $\pi\in\bool^{\leq D}$ and every input $x\in\bool^{n/2}$ for Alice (resp. Bob). The total running time is $O\left(D\cdot t\cprn{n/2,D}\cdot2^{n}\cdot2^D\right)$.
	\end{proof}
	
	For a protocol $\Pi$, we denote by $\mathrm{Leaves}(\Pi)$ the set of full transcripts of $\Pi$.\\
	
	\noindent \textbf{Remark.} We note that, in the \emph{white-box} context of the satisfiability problem, there is no need to assume a canonical partition of the input variables among the  players. For instance, a helpful partition can either be given as part of the input, or computed by the algorithm. As a consequence, in instantiations of Theorem \ref{thm:main_sat} for a particular circuit class $\mathcal{C}$, it is sufficient to be able to convert the input circuit from $\mathcal{C}$ into some device from $\FORMULA \circ \mathcal{G}$ for which protocols of bounded communication cost can be described.

    \subsection{Explicit approximating polynomials for formulas}
    
    From \Cref{thm:approx-polynomial}, we know that every size-$s$ formula has a degree-$O(\sqrt{s})$ polynomial that point-wisely approximates it. In our SAT algorithms, we will need to \emph{explicitly construct} such an approximating polynomial given a formula. One way to do this is to use an \emph{efficient} quantum query algorithm for formulas. It is known that a quantum query algorithm for a function $f$ using at most $T$ queries implies an approximating polynomial for $f$ of degree at most $2T$~\cite{BBC+01}, and by classically simulating such an quantum algorithm, one can show that the approximating polynomial can be obtained in time that is polynomial in the number of its monomials, in addition to the time for the classical simulation. For our task, we can use the result of Reichardt~\cite{Rei11a} which showed an \emph{efficient} quantum algorithm for evaluating size-$s$ formulas with $\bigo{\sqrt{s}\cdot \log s}$ queries\footnote{It is also known that there exists a quantum query algorithm for evaluating size-$s$ formulas with $\bigo{\sqrt{s}}$ queries~\cite{Rei11b}, which implies the existence of an approximating polynomial for size-$s$ formulas of degree $\bigo{\sqrt{s}}$ (see \Cref{thm:approx-polynomial}). However, because this algorithm is not known to be efficient, it is unclear whether such an approximating polynomial can be constructed efficiently with respect to the number of monomials.}. Here, we present an alternate way to construct approximating polynomials for de Morgan formulas which rely only on the \emph{existence} of such polynomials, without requiring an efficient quantum query algorithm. This ``black-box'' approach was suggested to us by an anonymous reviewer.
    
    We first need the following structural lemma for formulas.
    
    \begin{lemma}[\cite{IMZ12,Tal14}]\label{lemma:sat-decomposing}
        For every integer $s>0$, there exists an algorithm such that given a size-$s$ de Morgan formula $F$, runs in $\poly(s)$ time and outputs a top formula $F'$ with $O(\sqrt{s})$ leaves and each leaf of $F'$ is a sub-formula with $O(\sqrt{s})$ input leaves.
    \end{lemma}
    
    \begin{lemma}\label{thm:sat-explicit-polynomial}
		For any integer $s>0$ and any $0<\varepsilon<1$, there exists an algorithm of running time $s^{O\left(\sqrt{s}\cdot\log(s)\cdot\log(1/\varepsilon)\right)}$ such that given a de Morgan formula $F$ of size $s$, outputs an $\varepsilon$-approximating polynomial of degree $O(\sqrt{s}\cdot\log(s)\cdot\log(1/\varepsilon))$ for $F$. That is, the algorithm outputs a multi-linear polynomial (as sum of monomials) over the reals such that for every $x\in\bool^n$,
		\[
			\left| p(x)-F(x)\right| \leq \varepsilon.
		\]
	\end{lemma}
	
	\begin{proof}
	    We first note that it suffices to construct a $(1/3)$-approximating polynomial for $F$ with degree $D=O(\sqrt{s}\cdot\log(s))$. This is because given a $(1/3)$-approximating polynomial one can obtain explicitly an $\varepsilon$-approximating polynomial of degree $D\cdot O(\log(1/\varepsilon))$, by feeding $O(1/\varepsilon)$ copies of the $(1/3)$-approximating polynomial to the polynomial computing MAJORITY on $O(1/\varepsilon)$ bits~\cite{BNRdW07} (see also~{\cite[Appendix B]{Tal14}}).
	    
	    We first invoke \Cref{lemma:sat-decomposing} on $F$ to obtain a top formula $F'$ with $t=O(\sqrt{s})$ leaves, each of which is a sub-formula of size $O(\sqrt{s})$. We construct a $(1/20)$-approximating (multi-linear) polynomial $P$ for the top formula $F'$, which has degree $d_1=O(s^{1/4})$ by \Cref{thm:approx-polynomial}. Note that $P$ can be constructed in time $2^{O(\sqrt{s})}$ because $F'$ has at most $O(\sqrt{s})$ leaves. Next, for each of the $t$ sub-formulas, denoted as $F_1,F_2,\dots,F_t$, we construct a $(1/(20t))$-approximating polynomial. Note that these polynomials have degree $d_1=O(s^{1/4}\cdot\log(s))$ and can be constructed in time $2^{O(\sqrt{s})}$. Let's denote these $t$ polynomials as $Q_1, Q_2,\dots,Q_{t}$. Now for each $Q_i$ ($i\in[t]$), we define
	    \[
	        q_i(x)=\frac{Q_i(x)+1/(20t)}{1+1/(10t)}.
	    \]
	    The final approximating polynomial for $F$ is given as
	    \[
	        p(x)=P\left(q_1(x),q_2(x),\dots,q_t(x)\right).
	    \]
	    Note that $p$ has degree $d_1\cdot d_2=O(\sqrt{s}\cdot\log(s))$ and can be constructed (as sum of monomials) in time $s^{O(\sqrt{s}\cdot\log(s))}$. It remains to show that $p$ $(1/3)$-approximates $F$.
	    
	    For $0\leq q\leq 1$, let $N_q$ be the distribution over $\bool$ such that $\Prob_{y\sim N_q}[y=1]=q$. Then for an fixed input $x\in\bool^s$, we have
	    \begin{equation}\label{sat-bb1}
	        p(x)=\Exp_{y_i\sim N_{q_i(x)}}[P(y_1,y_2,\dots,y_t)].
	    \end{equation}
        Let $\mathcal{E}$ be the event that $y_i=F_i(x)$ for all $i\in[t]$. Note that
        \begin{equation}\label{sat-bb2}
            \delta \vcentcolon=\Prob_{y_i\sim N_{q_i(x)}}[\neg \mathcal{E}]\leq 1/10.
        \end{equation}
        To see \Cref{sat-bb2}, note that for every $i\in[t]$, if $F_i(x)=0$, then $0\leq q_i(x)\leq 1/(10t)$, which implies
        \[
            \Prob_{y_i\sim N_{q_i(x)}}[y_i\neq F_i(x)]\leq 1/(10t).
        \]
        Similar for the case when $F_i(x)=1$ (which implies $1-1/(10t)< q_i(x)\leq 1$). Then \Cref{sat-bb2} follows from a union bound. Now we can re-write \Cref{sat-bb1} as
        \begin{align*}
            p(x) &= \Exp[P(y_1,y_2,\dots,y_t)\mid \mathcal{E}]\cdot \Prob[\mathcal{E}] + \Exp[P(y_1,y_2,\dots,y_t)\mid \neg\mathcal{E}]\cdot \Prob[\neg\mathcal{E}]\\
            &=\left(F'(F_1(x),F_2(x),\cdot F_t(x))\pm 1/20 \right)\cdot(1-\delta)+ \Exp[P(y_1,y_2,\dots,y_t)\mid \neg\mathcal{E}]\cdot\delta.
        \end{align*}
        Note that $P(y)\in[-1/(20t),1+1/(20t)]$ for every $y\in\bool^t$, and that $\delta \leq 1/10$. A simple calculation shows that
        \[
               p(x) = F'(F_1(x),F_2(x),\dots, F_t(x))\pm \frac{1}{3},
        \]
        as desired.
	\end{proof}

	\subsection{The \#SAT algorithm}
	
	In this subsection, we present our \#SAT algorithm.
	
	\begin{theorem}\label{thm:sat-main}
	    For any integer $s>0$, there exists a deterministic \#SAT algorithm for $\FORMULA[s]\circ\mathcal{G}$, where $\mathcal{G}$ is the class of functions with explicit two-party deterministic protocols of communication cost at most $D$, that runs in time 
		\[
	        2^{n-\frac{n}{\sqrt{s}\cdot\log^2(s)\cdot D}}.
		\]
		In the case $\mathcal{G}$ is the class of functions with explicit randomized protocols of communication cost at most $R$, there exists an analogous randomized algorithm with a running time
		\[
	        2^{n-\left(\frac{n}{\sqrt{s}\cdot\log^2(s)\cdot R}\right)^{1/2}}.
	    \]
	\end{theorem}
	
	The algorithm is based on the framework for designing satisfiability algorithms developed by Williams~\cite{Wil14}. The idea is to transform a given circuit into a ``sparse polynomial'' and solve satisfiability by evaluating the polynomial on all points in a faster-than-brute-force manner.
	
	We first need the following fast matrix multiplication algorithm for ``narrow'' matrices.
	
	\begin{theorem}[\cite{Cop82}]\label{thm:sat-fast-matrix-mul}
		Multiplication of an $N\times N^{.172}$ matrix with an $N^{.172}\times N$ matrix can be done in $O(N^2 \log^2 N)$ arithmetic operations over any field.
	\end{theorem}

	For an even number $n>0$, and $x\in\bool^{n}$, we denote by $\LH{x}$ (resp. $\RH{x}$) the first half of $x$ and $\RH{x}\in\bool^{n/2}$ the second half. We now prove \Cref{thm:sat-main}.
	
	\begin{proof}[Proof of \Cref{thm:sat-main}]
    	We first prove the deterministic case.
    	
    	Let $C=F(g^1,g^2\dots,g^s)$ be a device in $\FORMULA\circ\mathcal{G}$ where $F$ is a formula and $g^1,g^2,\dots,g^s$ are functions that have a explicit communication protocol of cost at most $D$.
        The first step is to output the protocol tree for each $g^i$ ($i\in[s]$). Since each $g^i$ has explicit protocol of cost at most $D$, by \Cref{lem:sat-time-efficient}, these protocol trees can be output in time $s\cdot 2^{n/2+D+o(n)}\leq 2^{n/1.9}$ (here we assume $D = o(n)$ otherwise the theorem holds trivially).
    
    	Let $n'$ be an integer whose value is determined later. Let $T$ be a set of $n'$ variables such that $T$ contains $n'/2$ variables from the first half of the $n$ variables and the rest are from the second half. For a partial assignment $z\in\bool^{n'}$ to $T$, denote by $C_z$ the restricted function of $C$ where the variables in $T$ are fixed according to $z$. To count the number of satisfying assignments of $C$, we need to compute the following quantity:
		\begin{equation}\label{eq:sat-1}
			\sum_{x\in\bool^{n-n'}} \sum_{z\in\bool^{n'}} C_z(x).
		\end{equation}
		Now consider
		\[
			Q(x) = \sum_{z\in\bool^{n'}} C_z(x).
		\]
		We will try to obtain the value of $Q(x)$ for every $x\in\bool^{n-n'}$, in time about $2^{n-n'}$, which will allow us to compute the quantity in \Cref{eq:sat-1} in time $O(2^{n-n'})$ by summing $Q(x)$ over all the $x$'s. We do this by first transforming $Q$ into an \emph{approximating} polynomial with not-too-many monomials, and each monomial is a product of \emph{functions that only rely on either the first or the second half of $x$}. With such a polynomial, we can perform fast multipoint evaluation using the fast matrix multiplication algorithm in \Cref{thm:sat-fast-matrix-mul}.

		For each $z\in\bool^{n'}$, we view the formula $C_z$ as $F(g^1_z,g^2_z,\dots,g^s_z)$, where $F$ is the de Morgan formula part of $C_z$ and $g^1_z,g^2_z,\dots,g^s_z$ are the leaf gates. Let us now replace $F$ by a $\varepsilon$-approximating polynomial $p$, where $\varepsilon = 1/\left(3\cdot 2^{n'}\right)$, using \Cref{thm:sat-explicit-polynomial}. Note that the degree of $p$ is at most 
		\[
		    d\leq O(\sqrt{s}\cdot\log(s)\cdot\log(1/\varepsilon))\leq O(\sqrt{s}\cdot\log(s)\cdot n').
		\]
		Now consider the following
		\[
			Q'(x)= \sum_{z\in\bool^{n'}} p(g^1_z(x),g^2_z(x),\dots,g^s_z(x)).
		\]
		First, note that by the value that we've chosen for the approximating error $\varepsilon$, we have that, for every $x$,
		\[
			\left|Q'(x)-Q(x)\right|\leq 2^{n'}\cdot \varepsilon= 1/3.
		\]
		In other words, given $Q'(x)$, we can recover the value of $Q(x)$, which is supposed to be an integer.

		Next, we perform fast multipoint evaluation on $Q'$. First of all, we re-write $Q'$ as follows:
		\begin{equation}\label{eq:sat-2}
			Q'(x) = \sum_{z\in\bool^{n'}} \sum_{\substack{S\subseteq[s]: \\ |S|\leq d}}\hat{p}(S)\cdot\prod_{i\in S}g^i_z(x).
		\end{equation}
		Now let $\Pi_i$ be the protocol of $g^i$, we can re-write $g^i_z$ as follows:
		\begin{equation}\label{eq:sat-3}
		    g^i_z(x)=\sum_{\pi_i\in\mathrm{Leaves}(\Pi_i)} \alpha^i\left(\LH{z}\LH{x},\pi_i\right)\cdot \beta^{i}\left(\RH{z}\RH{x},\pi_i\right),
		\end{equation}
		where $\alpha^i\left(\LH{z}\LH{x},\pi_i\right)$ (resp. $\beta^{i}\left(\RH{z}\RH{x},\pi_i\right)$) is $1$ if and only if $\left(\LH{z}\LH{x}\right)$ (resp. $\left(\RH{z}\RH{x}\right)$) belongs to the rectangle indexed by $\pi_i$ and the function value of that rectangle is $1$. Note that for each $i\in[s]$, given the pre-computed protocol tree of the $\Pi_i$, $\alpha^{i}$ and $\beta^{i}$ can be computed in polynomial time (for example, using binary search).
		After plugging \Cref{eq:sat-3} into \Cref{eq:sat-2} for every $i\in[s]$ and rearranging, we get
		\begin{equation}\label{eq:sat-4}
		    Q'(x) = \sum_{z\in\bool^{n'}} \sum_{\substack{S\subseteq[s]: \\ |S|\leq d}} \sum_{\substack{\vec{\pi}=(\pi_i)_{i\in S}: \\ \pi_i\in\mathrm{Leaves}(\Pi_i)}} \hat{p}(S)\cdot \prod_{i\in S} \alpha^i\left(\LH{z}\LH{x},\pi_i\right)\cdot \prod_{i\in S} \beta^{i}\left(\RH{z}\RH{x},\pi_i\right).
		\end{equation}
		Note that $Q'$ can be expressed as the sum of at most $m$ terms, where
		\[
			m\leq 2^{n'}\cdot s^{O(\sqrt{s}\cdot \log(s)\cdot n')}\cdot 2^{O(\sqrt{s}\cdot \log(s)\cdot n'\cdot D)}\leq 2^{O(\sqrt{s}\cdot\log^2(s)\cdot D\cdot n')}.
		\]
		Note that given \Cref{thm:sat-explicit-polynomial}, we can obtain $Q'$ in time
		\begin{equation}\label{eq:sat-time-1}
		    2^{O(\sqrt{s}\cdot\log^2(s)\cdot D\cdot n')}.
		\end{equation}
		Next, we construct a $2^{(n-n')/2} \times m$ matrix $A$ and a $m \times 2^{(n-n')/2}$ matrix $B$ as follows:
		\[
			A_{\LH{x},(z,S,\vec{\pi})}=\hat{p}(S)\cdot\prod_{i\in S} \alpha^i\left(\LH{z}\LH{x},\pi_i\right),
		\]
		and 
		\[
			B_{(z,S,\vec{\pi}),\RH{x}}=\prod_{i\in S} \beta^{i}\left(\RH{z}\RH{x},\pi_i\right).
		\]
		It is easy to see that for each $x\in\bool^{n-n'}$,
		\[
			Q'(x) = (A\cdot B)_{\LH{x},\RH{x}}.
		\]
		We now want to compute $A\cdot B$. Therefore, we want $m \leq 2^{.172(n-n')/2}$ so that computing $A\cdot B$ can be done in time $\tilde{O}(2^{n-n'})$ using \Cref{thm:sat-fast-matrix-mul}. For this we can set $n'$ to be
		\[
			n'=\frac{n}{c\cdot\sqrt{s}\cdot \log^{2}(s)\cdot D},
		\]
		where $c>0$ is some sufficiently large constant. Together with the running time in \Cref{eq:sat-time-1}, The total running time of the algorithm is therefore
		\[
		    2^{n-\frac{n}{\sqrt{s}\cdot\log^2(s)\cdot D}}.
		\]
		
		For the randomized case, for each $g^i$ ($i\in[s]$), we consider a randomized protocol $\Pi_i$ that has error $\varepsilon'\leq 1/(3\cdot s\cdot2^{n'})$, and replace $g^i$ with a randomly picked protocol from $\Pi_i$, so we can say that for every $x\in{n-n'}$, the algorithm computes $Q(x)$ (or $Q'(x)$) with probability at least $2/3$ (via a union bound over all the $g^i$'s and a union bound over all the $z$'s in $\bool^{n'}$). Then we can repeat the above algorithm $\poly(n)$ times and obtain $Q(x)$ for all $x\in\bool^{n-n'}$ correctly with high probability. Note that the error of any randomized protocol with communication complexity $R$ can be reduced to $\varepsilon'$ by blowing up the communication complexity by a factor of $O(\log(1/\varepsilon'))$. In this case the, (as we are considering longer transcripts) the number of terms in $Q'$ (as in \Cref{eq:sat-4}) will be
		\[
		    2^{O(\sqrt{s}\cdot\log^2(s)\cdot R\cdot (n')^2)},
		\]
		and we need to set accordingly
		\[
		    n'=\Omega\left(\frac{n}{\sqrt{s}\cdot\log^2(s)\cdot R}\right)^{1/2},
		\]
		which gives the claimed running time for the randomized case.
	\end{proof}
	
	In fact, using the ideas above we can also get a randomized \#SAT algorithm for the more expressive class $\FORMULA \circ\AC^{0}_{d,M}\circ\mathcal{G}$, where $\AC^{0}_{d,M}$ is the class of depth-$d$ size-$M$ circuits and $\mathcal{G}$ is the class of functions that have low-communication complexity\footnote{Here we define the size of a $\AC^{0}_{d,M}$ circuit to be the number of wires. Note that a circuit in $\FORMULA \circ\AC^{0}_{d,M}\circ\mathcal{G}$ can have $M$ functions from $\mathcal{G}$ at the bottom.}, by combining with the fact that $\AC^0$ circuits have low-degree \emph{probabilistic polynomials over the reals} (a probabilistic polynomial of a function $f$ is a distribution on polynomials such that for every input $x$, a randomly picked polynomial from the distribution agrees with $f$ on the input $x$). More specifically, we have the following.
	
	\begin{theorem}\label{thm:sat-formula-ac0-lc}
	    For any integers $s,d,M>0$, there exists a randomized \#SAT algorithm for $\FORMULA[s]\circ\AC^{0}_{d,M}\circ\mathcal{G}$, where $\mathcal{G}$ is the class of functions with explicit two-party deterministic protocols of communication cost at most $D$, the algorithm outputs the number of satisfying assignments in time 
    	\[
    	    2^{n-\left(\frac{n}{\sqrt{s}\cdot\log^2(s)\cdot (\log M)^{O(d)} \cdot D}\right)^{1/2}}.
    	\]
		In the case $\mathcal{G}$ is the class of functions with explicit randomized protocols of communication cost at most $R$, there exists an analogous randomized algorithm with a running time
    	\[
    	    2^{n-\left(\frac{n}{\sqrt{s}\cdot\log^2(s)\cdot (\log M)^{O(d)} \cdot R}\right)^{1/3}}.
    	\]
	\end{theorem}
	
	\begin{proof}[Proof sketch]
	    We show the case where $\mathcal{G}$ has low randomized communication complexity. 
	    Let
	    \begin{itemize}
	        \item $\varepsilon_1 = 1/\left(3\cdot2^{n'}\right)$,
	        \item $\varepsilon_2 = 1/\left(6\cdot s\cdot2^{n'}\right)$ and
	        \item $\varepsilon_3 = 1/\left(6\cdot M\cdot2^{n'}\right)$.
	    \end{itemize}
	    As in the proof of \Cref{thm:sat-main}, we can replace the formula part of $\FORMULA[s] \circ\AC^{0}_{d,M}\circ\mathcal{G}$ with a  $\varepsilon_1$-approximating polynomial of degree
	    \[
	        O(\sqrt{s}\cdot\log(s)\cdot\log(1/\varepsilon_1))= O(\sqrt{s}\cdot\log(s)\cdot n').
	    \]
	    Then we replace the $\AC^{0}_{d,M}$ circuit with a randomly picked polynomial from a $\varepsilon_2$-error probabilistic polynomial. By~\cite{HS19}, such a probabilistic polynomial is constructive and has degree at most
	    \[
	        (\log M)^{O(d)}\cdot \log(1/\varepsilon_2)= (\log M)^{O(d)}\cdot (n' + \log(s)).
	    \]
	    Finally, we replace each of the bottom functions, which is from $\mathcal{G}$, with a randomly picked protocol from a randomized protocol with error $\varepsilon_3$, and hence has cost at most
	    \[ 
	        R\cdot O(\log(1/\varepsilon_3))= O(R\cdot(n'+\log(M))).
	    \]
	    As a result, we can express $Q'$ as a polynomial with at most 
    	\[
    		2^{O\left(\sqrt{s}\cdot\log^2(s)\cdot (\log M)^{O(d)} \cdot R\cdot (n')^3\right)}
    	\]
    	monomials, whose variables are functions that depend on either the first half or the second half of $x$. Note that with our choices of $\varepsilon_2$ and $\varepsilon_3$, for every $x\in\bool^{n-n'}$, the algorithm computes $Q(x)$ correctly that with probability at least $2/3$ (by union bounds). By the same reasoning as in the proof of \Cref{thm:sat-main}, we get a randomized \#SAT algorithm with running time
    	\[
    	    2^{n-\left(\frac{n}{\sqrt{s}\cdot\log^2(s)\cdot (\log M)^{O(d)} \cdot R}\right)^{1/3}},
    	\]
    	as desired.
	\end{proof}
	
	It is worth noting that unlike \Cref{thm:sat-main}, the algorithm in \Cref{thm:sat-formula-ac0-lc} is \emph{randomized} even if $\mathcal{G}$ is the class of functions with low \emph{deterministic} communication complexity, because of  the use of probabilistic polynomials for the $\AC^0$ circuits.
	
    \section{Learning algorithms}\label{sec:learning}
	
    In this section, we prove the following learning result for the $\FORMULA \circ \mathsf{XOR}$ model.
    
    \begin{theorem}\label{thm:learn-formula-xor}
        For every constant $\gamma > 0$, there is an algorithm that \emph{PAC} learns the class of $n$-variate Boolean functions $\FORMULA[n^{2 - \gamma}]\circ\mathsf{XOR}$ to accuracy $\varepsilon$ and with confidence $\delta$ in time $\mathsf{poly}\big (2^{n/\log n}, 1/\varepsilon, \log(1/\delta) \big )$.
    \end{theorem}
    
    We first review some useful results that pertain to agnostically learning parities as well as boosting of learning algorithms.
    
    \subsection{Agnostically learning parities and boosting}
    
    For a parameter $n \geq 1$, let $\Delta$ be a distribution on labelled examples $\prn{x,y}$ supported over $\bool^n\times\bool$, and assume that for each $x$ there is at most one $y$ such that $(x,y) \in \mathsf{Support}(\Delta)$. For a function $h \colon \bool^n \to \bool$, we denote by $\err_\Delta\cprn{h}$ the error of $h$ under this distribution:
    \[
        \err_{\Delta}\cprn{h} \;=\; \Prob_{\prn{x,y}\sim \Delta}\csqbra{h\cprn{x}\neq y}.
    \]
    Similarly, for a class of functions $\mathcal{C}$, we let $\opt_\Delta\cprn{\mathcal{C}}$ be the error of the best function in the class:
    \[
        \opt_\Delta\cprn{\mathcal{C}} \;=\; \min_{h\in\mathcal{C}}\;\err_\Delta\cprn{h}.
    \]
    We will need a result established by Kalai, Mansour, and Verbin~\cite{KMV08}, which gives a non-trivial time agnostic learning algorithm for the class of parities.
    
    \begin{lemma}[\cite{KMV08}]\label{lem:parity-learner}
        Let $\mathsf{XOR}$ be the class of parity functions on $n$ variables. Then, for any constant $\zeta > 0$, there is a randomized learning algorithm $W$ such that, for every parameter $n\geq1$ and distribution $\Delta$ over labelled examples, when $W$ is given access to independent samples from $\Delta$ it outputs with high probability a circuit computing a hypothesis $h:\brkts{0,1}^n\to\brkts{0,1}$ such that
        \[
            \err_\Delta\cprn{h} \;\leq\; \opt_\Delta\cprn{\mathsf{XOR}}+2^{-n^{1-\zeta}}.
        \]
        The sample complexity and running time of $W$ is $2^{O\prn{n/\log n}}$.
    \end{lemma}

    Recall that a boosting procedure for learning algorithms transforms a weak learner that outputs a hypothesis that is just weakly correlated with the unknown function into a (strong) PAC learning algorithm for the same class (i.e., a learner in the sense of Definition \ref{def:pac}). We refer for instance to \citep{KV94} for more information about boosting in learning theory. We shall make use of the following boosting result by Freund~\cite{Fre90}.
    
    
    \begin{lemma}[\cite{Fre90}]\label{lem:freund-boosting}
        Let $W$ be a \emph{(}weak\emph{)} learner for a class $\mathcal{C}$ that runs in time $t(n)$ and outputs \emph{(}under any distribution\emph{)} a hypothesis of error up to $1/2-\beta$, for some constructive function $\beta(n) > 0$. Then, there exists a \emph{PAC} learning algorithm for $\mathcal{C}$ that runs in time $\mathsf{poly}(n,t,1/\varepsilon,1/\beta,\log (1/\delta))$.
    \end{lemma}

    \subsection{PAC-learning small formulas of parities}
    
    
    We are ready to show that sub-quadratic size formulas over parity functions can be learned in time $2^{O(n/\log n)}$. First, we argue that Lemma \ref{lem:parity-learner} provides a weak learner that works under any distribution $\mathcal{D}$ supported over $\{0,1\}^n$. This will follow from Lemma \ref{lem:lb-Tal-2}, which shows that any function in $\FORMULA[s] \circ \mathsf{XOR}$ is correlated with some parity function with respect to $\mathcal{D}$. We then obtain a standard PAC learner via the boosting procedure from Lemma \ref{lem:freund-boosting}.
    
    \begin{proof}[Proof of \Cref{thm:learn-formula-xor}]
        Let $\mathcal{C} = \mathsf{FORMULA} \circ \mathsf{XOR}$, where $s = n^{2 - \gamma}$ for some constant $\gamma > 0$. For any function $f \in \mathsf{FORMULA}[s] \circ \mathsf{XOR}$ and distribution $\mathcal{D}$ supported over $\{0,1\}^n$, Lemma \ref{lem:lb-Tal-2} shows that there exists a parity function $\chi = \chi(f,\mathcal{D})$ such that 
        $$
            \Pr_{x \sim \mathcal{D}}[f(x) = \chi(x)] \;\geq\;\frac{1}{2} + \frac{1}{2^{n^{1 - \lambda}}},
        $$
        for some $\lambda = \lambda(\gamma) > 0$ independent of $n$, under the assumption that $n$ is sufficiently large. Let $\Delta = \Delta(\mathcal{D},f)$ be the distribution over labelled examples induced by $\mathcal{D}$ and $f$. Note that $\opt_\Delta\cprn{\mathsf{XOR}} \leq 1/2 - \exp{\!(n^{1 - \lambda})}$. Consequently, by invoking Lemma \ref{lem:parity-learner} with parameter $\zeta = \lambda$, it follows that $\FORMULA[n^{2 - \gamma}] \circ \mathsf{XOR}$ can be learned under an arbitrary distribution to error $\beta(n) \leq 1/2 - \exp(n^{1 - \Omega(1)})$ in time $t(n) = 2^{O(n/\log n)}$. Consequently, we can obtain a PAC learner algorithm for $\FORMULA[n^{2 - \gamma}] \circ \mathsf{XOR}$ via Lemma \ref{lem:freund-boosting} that runs in time $\mathsf{poly}(n,t(n),1/\varepsilon,1/\beta,\log(1/\delta)) = \mathsf{poly}(2^{n/\log n}, 1/\varepsilon, \log(1/\delta))$.
    \end{proof}

	\section*{Acknowledgements}
	
	We would like to thank Rocco Servedio for bringing to our attention the work by Kalai, Mansour, and Verbin~\citep{KMV08}, which is a central ingredient in the proof of Theorem \ref{thm:main_learning}. We also thank Mahdi Cheraghchi for several discussions on the analysis of Boolean circuits with a bottom layer of parity gates.
	
	This work was funded in part by a Royal Society University Research Fellowship (URF$\setminus$R1$\setminus$191059).
	
	\bibliographystyle{alpha}
    \bibliography{main}
    
    \appendix

    \section{Proofs of useful lemmas}
    
    \subsection{Useful lemmas for formulas}\label{subsec:formula-lemmas}
    
    The proofs in this section are essentially the same as that of~\cite{Tal16}.
    
    \begin{lemma}[\cite{Tal16}, \Cref{lem:lb-Tal-1} restated]\label{lem:lb-Tal-1-app}
	    Let $\mathcal{D}$ be a distribution over $\boolpm^n$, and let $f,C\colon\boolpm^n\to\boolpm$ be such that
	    \[
	        \Prob_{x\sim\mathcal{D}}[C(x)=f(x)]\geq 1/2+\varepsilon.
	    \] 
	    Let $\tilde{C}\colon\boolpm^n\to\real$ be a $\varepsilon$-approximating function of $C$, i.e., for every $x\in\boolpm^n$, $|C(x)-\tilde{C}(x)|\leq \varepsilon$. Then,
	    \[
	        \Exp_{x\sim\mathcal{D}}[\tilde{C}(x)\cdot f(x)]\geq \varepsilon.
	    \]
	\end{lemma}
	
	\begin{proof}
		Note that since $\tilde{C}$ $\varepsilon$-approximate $C$, we have for every $x\in\boolpm^n$
		\[
			\tilde{C}\cdot C(x)\geq 1-\varepsilon,
		\]
		and 
		\[
			\tilde{C}\cdot (1-C(x))\geq -1-\varepsilon.
		\]
		Then,
		\begin{align*}
			\Exp_{x\sim\mathcal{D}}[\tilde{C}(x)\cdot f(x)]&=\Exp_{x\sim\mathcal{D}}[\tilde{C}(x)\cdot f(x)\given C(x)=f(x)]\cdot\Prob_{x\sim\mathcal{D}}[C(x)=f(x)]\\
			&\qquad+\Exp_{x\sim\mathcal{D}}[\tilde{C}(x)\cdot f(x)\given C(x)\ne f(x)]\cdot\Prob_{x\sim\mathcal{D}}[C(x)\ne f(x)]\\
			&\geq(1-\varepsilon)\cdot \Prob_{x\sim\mathcal{D}}[C(x)=f(x)] + (-1-\varepsilon)\cdot \left(1-\Prob_{x\sim\mathcal{D}}[C(x)=f(x)]\right)\\
			&=2\cdot \Prob_{x\sim\mathcal{D}}[C(x)=f(x)] - 1 -\varepsilon\\
			&\geq 2\cdot (1/2+\varepsilon) - 1 - \varepsilon \geq \varepsilon,
		\end{align*}
		as desired.
	\end{proof}
	
	\begin{lemma}[\cite{Tal16}, \Cref{lem:lb-Tal-2} restated]\label{lem:lb-Tal-2-app}
	    Let $\mathcal{D}$ be a distribution over $\boolpm^n$ and let $\mathcal{G}$ be a class of functions. For $f\colon\boolpm^n\to\boolpm$, suppose that $D\colon\boolpm^n\to\boolpm\in \FORMULA[s]\circ\mathcal{G}$ is such that 
	    \[
	        \Prob_{x\sim\mathcal{D}}[D(x)=f(x)]\geq 1/2+\varepsilon_0.
	    \]
	    Then there exists some $h\colon\boolpm^n\to\boolpm\in \XOR_{O\left(\sqrt{s}\cdot\log(1/\varepsilon_0)\right)}\circ \mathcal{G}$ such that
	    \[
	        \Exp_{x\sim\mathcal{D}}[h(x)\cdot f(x)]\geq \frac{1}{s^{O\left(\sqrt{s}\cdot\log(1/\varepsilon_0)\right)}}.
	    \]
	\end{lemma}
	
    \begin{proof}
		Let 
		\[
			D = F(g_1,g_2\dots,g_s)
		\]
		be a device in $\FORMULA\circ\mathcal{G}$ where $F$ is a formula and $g_1,g_2,\dots,g_s$ are function from $\mathcal{G}$.

		Let $p\colon\boolpm^s\to\real$ be a $\varepsilon_0$-approximating polynomial for $F$ of degree $d=O(\sqrt{s}\cdot\log(1/\varepsilon_0))$. Note that we can write
		\[
			p(z)=\sum_{\substack{S\subseteq[s]: \\ |S|\leq d}}\hat{p}(S)\cdot\prod_{i\in S}z_i.
		\]
		Also, for each $S\subseteq[s]$, we have
		\[
			|\hat{p}(S)|=\left| \Exp_{z\in\boolpm^s}[p(z)\cdot\prod_{i\in S}z_i] \right| \leq  1+\varepsilon_0.
		\]
		Now let
		\[
			\tilde{D} \vcentcolon= p(g_1,g_2\dots,g_s).
		\]
		Note that $\tilde{D}$ is a $\varepsilon_0$-approximating function for $D$. Therefore, by \Cref{lem:lb-Tal-1-app}, we have
		\begin{align*}
			\varepsilon_0 &\leq \Exp_{x\sim\mathcal{D}}[D(x)\cdot f(x)]\\
			&= \Exp_{x\sim\mathcal{D}}\left[\left(\sum_{\substack{S\subseteq[s]: \\ |S|\leq d}}\hat{p}(S)\cdot\prod_{i\in S}g_i \right) \cdot f(x)\right]\\
			&=\sum_{\substack{S\subseteq[s]: \\ |S|\leq d}}\hat{p}(S)\cdot \Exp_{x\sim\mathcal{D}}\left[\prod_{i\in S}g_i\cdot f(x)\right]\\
			&\leq \sum_{\substack{S\subseteq[s]: \\ |S|\leq d}} (1+\varepsilon_0) \cdot \left| \Exp_{x\sim\mathcal{D}}\left[\prod_{i\in S}g_i\cdot f(x)\right] \right|.
		\end{align*}
		The above equation is the sum of at most $s^{O(d)}$ summands. Therefore, there exists some $S\subseteq[s]$ such that
		\[
			\left|\Exp_{x\sim\mathcal{D}}\left[\prod_{i\in S}g_i\cdot f(x)\right]\right| \geq \frac{\varepsilon_0}{(1+\varepsilon_0)\cdot s^{O(d)}}\geq \frac{1}{s^{O\left(\sqrt{s}\cdot\log(1/\varepsilon_0)\right)}},
		\]
		which implies that there exists some $h$, such that either $h=\prod_{i\in S}g_i$ or $h=-\prod_{i\in S}g_i$, and
		\[
			\Exp_{x\sim\mathcal{D}}\left[h(x)\cdot f(x)\right] \geq \frac{1}{s^{O\left(\sqrt{s}\cdot\log(1/\varepsilon_0)\right)}}.
		\]
		Finally, note that such $h$ can be expressed as the $\XOR$ of at most $d$ functions from $\mathcal{G}$.
    \end{proof}
    
    \subsection{PRG for low-communication functions in the number-in-hand setting}\label{appendix:prg_cc}
	
	In this subsection, we show how to fool functions with low communication complexity in the number-in-hand model.
    
	\begin{theorem}[\cite{ASWZ96,INW94}, \Cref{thm:prg-lc} restated]\label{lem:prg-lc-app}
		For any $k\geq 2$, there exists a PRG that $\delta$-fools any $n$-bits functions with $k$-party number-in-hand deterministic communication complexity at most $D'$, with seed length
		\[
			n/k+O\left(D' + \log(1/\delta)+\log(k)\right)\cdot\log(k).
		\]
	\end{theorem}
	
	The PRG in \Cref{thm:prg-lc} is based on the PRG by Impagliazzo, Nisan and Wigderson~\cite{INW94} that is used to derandomize ``network algorithms" and space-bounded computation. We will need to use randomness extractors, which we review below.
	
	\begin{definition}[Min-entropy]
    	Let $X$ be a random variable. The \emph{min-entropy} of $X$, denoted by $H_\infty(X)$, is the largest real number $k$ such that $\Prob\csqbra{X=x}\leq2^{-k}$ for every $x$ in the range of $X$. If $X$ is a distribution over $\boolpm^\EI$ with $H_\infty\cprn{X}\geq k$, then $X$ is called a {\em$\prn{\EI,k}$-source}.
    \end{definition}
    
    \begin{definition}[Extractors]
    	A function $\mathrm{Ext}\colon\boolpm^\EI\times\boolpm^d\to\boolpm^m$ is an \emph{$(k,\varepsilon)$-extractor} if, for any $\prn{\EI,k}$-source $X$, and any test $T\colon\boolpm^m\to\boolpm$, it is the case that
    	\[
    	    \left|\Prob[T(\mathrm{Ext}\cprn{X,\uniform_d}X)=1]-\Prob[T(\uniform_m)=1]\right|\leq \varepsilon.
    	\]
    \end{definition}
	
    \begin{theorem}[{\cite[Theorem 6.22]{Vad12}}]\label{thm:prg-ext}
		For any integer $m,\kappa>0$ and $0<\delta'<0$, there exists an explicit $(\kappa,\delta')$ extractor $\mathrm{Ext}\colon\bool^{m}\times\bool^d\to\bool^{m}$ with $d=O(m-k+\log(1/\delta'))$.
	\end{theorem}

    We are now ready to show \Cref{thm:prg-lc}.
    
	\begin{proof}[Proof of \Cref{thm:prg-lc}]\renewcommand{\qedsymbol}{$\square$ (\Cref{thm:prg-lc})}
	    We first describe the construction of the PRG. In fact, we will construct a sequence of PRGs $G_0,G_1,\dots,G_{\log(k)}$. We begin by specifying the parameters of these PRGs.
	    Let $t=\log(k)$, and let
	    \[
	    	d=O\left(D'+\log(1/\delta)+t\right).
	    \]
	    For $i=0,1,\dots,t$, let
	    \begin{itemize}
	        \item $r_0=n/k$,
	        \item $r_i=r_{i-1}+d$.
	    \end{itemize}
	    Note that we have $r_i=n/k+i\cdot d$. Also, let
	    \[
	    	\mathrm{Ext}_i\colon\bool^{r_i}\times\bool^d\to\bool^{r_i}
	    \]
	    be a $(\kappa_i,\delta')$-extractor from \Cref{thm:prg-ext}, where
	    \[
	    	\kappa_i=r_i-D'-2t - \log(1/\delta)
	    \] 
	    and
	    \[
	        \delta'=\delta/\left(3^t\cdot 2^{D'}\right).
	    \]
	    Note that the seed length of the extractors is $d=O\left(D'+\log(1/\delta)+t\right)$. Finally, define $G_i\colon\bool^{r_i}\to\bool^{n/2^{t-i}}$ recursively as follows
	    \begin{itemize}
	    	\item $G_{0}(a)=a$, where $a\in\bool^{n/k}$.
	    	\item $G_i(a,z)=G_{i-1}(a)\circ G_{i-1}(\mathrm{Ext}_{i-1}(a,z))$, where $a\in\bool^{r_{i-1}}$ and $z\in\bool^d$.
	    \end{itemize}
	    We will show that $G_t\colon\bool^{r_t=n/k+t\cdot d}\to\bool^n$ fools any functions $f$ with $k$-party number-in-hand deterministic communication complexity at most $D'$.
	    First, note that such $f$ can be written as
	    \[
	    	f(x_1,x_2,\dots,x_k)=\sum_{i=1}^{2^{D'}} h^{(i)}_{1}(x_1)\cdot h^{(i)}_{2}(x_2)\cdot\ldots\cdot h^{(i)}_{k}(x_k),
	    \]
	    for some $h^{(i)}_{j}\colon\bool^{n/k}\to \bool$ ($i\in\left[2^{D'}\right], j\in[k]$).
	    Therefore, to show that the PRG $G_t$ $\delta$-fool $f$, it suffices to show that $G_t$ $\left(\delta/2^{D'}\right)$-fools every function $g$ of the form
	    \[
	    	g(x_1,x_2,\dots,x_k)=h_1(x_1)\cdot h_2(x_2)\cdot\ldots\cdot h_k(x_k).
	    \]
	    More specifically we show the following.
	    
	    \begin{claim}\label{lem:prg-lc-claim-1}
	    	For every $k \geq 2$ and $  0\leq i \leq t$, the generator $G_i$ defined above $\left(3^{i}\cdot \delta'\right)$-fools every function $g_i\colon\bool^{n/2^{t-i}}\to\bool$ of the form
	    	\[
	    		g_i(x_1,x_2,\dots,x_{k/2^{t-i}})=h_1(x_1)\cdot h_2(x_2)\cdot\ldots\cdot h_{k/2^{t-i}}(x_{k/2^{t-i}}),
	    	\]
	    	where $x_1,x_2,\dots,x_{k/2^{t-i}}\in\bool^{n/k}$.
	    \end{claim}
	    
	    \begin{proof}\renewcommand{\qedsymbol}{$\square$ (\Cref{lem:prg-lc-claim-1})}
	    	The proof is by induction on $i$. The base case is $i=0$, which is trivial given the definition of $G_{0}$. Now suppose the claim holds for $i-1$, we show the case for $i$. This is done using a hybrid argument. Consider the following four distributions
	    	\begin{itemize}
	    		\item $\mathcal{D}_1=U_{n/2^{t-i}}$,
	    		\item $\mathcal{D}_2=U_{n/2^{t-i+1}}\circ G_{i-1}(U_{r_{i-1}})$,
	    		\item $\mathcal{D}_3=G_{i-1}(U_{r_{i-1}}) \circ G_{i-1}(U'_{r_{i-1}})$ ($U$ and $U'$ are two independent uniform distributions),
	    		\item $\mathcal{D}_4=G_i(U_{r_i})$.
	    	\end{itemize}
	    	We want show show that
	    	\[
	    		\left|\Exp[g_i(\mathcal{D}_1)]-\Exp[g_i(\mathcal{D}_4)]\right|\leq 3^{i}\cdot \delta'.
	    	\]
	    	By the triangle inequality, it suffices to show that
	    	\begin{equation}\label{eq:prg-lc-triangle}
	    		\left|\Exp[g_i(\mathcal{D}_1)]-\Exp[g_i(\mathcal{D}_2)]\right| + \left|\Exp[g_i(\mathcal{D}_2)]-\Exp[g_i(\mathcal{D}_3)]\right| + \left|\Exp[g_i(\mathcal{D}_3)]-\Exp[g_i(\mathcal{D}_4)]\right| \leq  3^{i}\cdot \delta'.
	    	\end{equation}
	    	We show \Cref{eq:prg-lc-triangle} by upper bounding each of the three summands.
	    	
	    	\paragraph{First summand.}
	    	
	    	We show that
            \begin{equation}\label{eq:prg-lc-triangle-1}
				\left|\Exp[g_i(\mathcal{D}_1)]-\Exp[g_i(\mathcal{D}_2)]\right| \leq 3^{i-1}\cdot \delta'.
			\end{equation}
			Let us re-write $g_i$ as
	    	\[
	    		g_i(x_1,x_2,\dots,x_{k/2^{t-i}})=\LH{h}(x_1,x_2,\dots,x_{k/2^{t-i+1}})\cdot \RH{h}(x_{k/2^{t-i+1}+1},x_{k/2^{t-i+1}+2},\dots,x_{k/2^{t-i}}),
	    	\]
	    	where
	    	\[
	    	    \LH{h}(y)\vcentcolon=\prod_{j=1}^{k/2^{t-i+1}}h_i(y) \text{\quad and\quad  } \RH{h}(y)\vcentcolon=\prod_{j=k/2^{t-i+1}}^{k/2^{t-1}}h_i(y).
	    	\]
	    	Then,
	    	\begin{align*}
	    		\Exp[g_i(\mathcal{D}_2)]&=\Exp\left[\LH{h}(U_{n/2^{t-i+1}})\cdot \RH{h}(G_{i-1}(U_{r_{i-1}}))\right]\\
	    		&=\Exp\left[\LH{h}(U_{n/2^{t-i+1}})\right]\cdot\Exp\left[\RH{h}(G_{i-1}(U_{r_{i-1}}))\right]\\
	    		&=\Exp\left[\LH{h}(U_{n/2^{t-i+1}})\right]\cdot\left(\Exp\left[\RH{h}(U_{n/2^{t-i+1}})\right]\pm 3^{i-1}\cdot \delta'\right)
	    		\tag{By the induction hypothesis}\\
	    		&=\Exp\left[\LH{h}(U_{n/2^{t-i+1}})\right]\cdot\Exp\left[\RH{h}(U_{n/2^{t-i+1}})\right] \pm 3^{i-1}\cdot \delta'\\
	    		&=\Exp[g_i(\mathcal{D}_1)]\pm 3^{i-1}\cdot \delta',
	    	\end{align*}
	    	as desired.
	    	
	    	\paragraph{Second summand.}
	    	
	    	By a similar argument, it can be shown that
            \begin{equation}\label{eq:prg-lc-triangle-2}
				\left|\Exp[g_i(\mathcal{D}_2)]-\Exp[g_i(\mathcal{D}_3)]\right| \leq 3^{i-1}\cdot \delta'.
			\end{equation}
			We omit the details here.
	    	
	    	\paragraph{Third summand.}
	    	
	    	We show that
            \begin{equation}\label{eq:prg-lc-triangle-3}
				\left|\Exp[g_i(\mathcal{D}_3)]-\Exp[g_i(\mathcal{D}_4)]\right| \leq \delta'.
			\end{equation}
	    	We have
	    	\begin{align*}
	    		\Exp[g_i(\mathcal{D}_4)]&=\Exp[g_i(G_i(U_{r_i}))]\\
	    		&=\Exp\left[\LH{h}(G_{i-1}(X))\cdot \RH{h}(G_{i-1}(\mathrm{Ext}_{i-1}(X,Z)))\right]
	    		\tag{where $X\sim\bool^{r_{i-1}}$ and $Z\sim\bool^{d}$}\\
	    		&=\Exp[A(X)\cdot B(\mathrm{Ext}_{i-1}(X,Z))]
	    		\tag{where $A(\cdot)= \LH{h}(G_{i-1}(\cdot))$  and $B(\cdot)= \RH{h}(G_{i-1}(\cdot))$}\\
	    		&=\Exp[B(\mathrm{Ext}_{i-1}(X,Z))\given A(X)=1]\cdot\Prob[A(X)=1].
	    	\end{align*}
	    	Similarly, we get
	    	\[
	    		\Exp[g_i(\mathcal{D}_3)]=\Exp[B(U_{r_{i-1}})\given A(X)=1]\cdot\Prob[A(X)=1].
	    	\]
	    	As a result, we have
	    	\begin{align*}\label{eq:prg-lc-triangle-3.1}
	    		&\quad\left|\Exp[g_i(\mathcal{D}_4)]-\Exp[g_i(\mathcal{D}_3)]\right|\\
	    		&=\left|\left(\Exp[B(\mathrm{Ext}_{i-1}(X,Z))\given A(X)=1]-\Exp[B(U_{r_{i-1}})\given A(X)=1]\right)\cdot\Prob[A(X)=1]\right|.
	    		\numberthis
	    	\end{align*}
		    On the one hand, if $\Prob[A(X)=1]\leq \delta'$, then \Cref{eq:prg-lc-triangle-3.1} is at most $\delta'$. On the other hand, if $\Prob[A(X)=1]> \delta'$, then
	        \[
	        	H_{\infty}\cprn{X\given A(X)=1} > r_{i-1}-\log(1/\delta')> r_{i-1}-D'-2t -\log(1/\delta)=\kappa_{i-1}.
	        \]
	        Then by the fact that $\mathrm{Ext}_{i-1}$ is a $(\kappa_{i-1},\delta')$-extractor, we have
	        \[
	            \left|\Exp[B(\mathrm{Ext}_{i-1}(X,Z))\given A(X)=1] - \Exp[B(U_{r_{i-1}})\given A(X)=1]\right| \leq \delta'.
	        \]
	        Therefore, \Cref{eq:prg-lc-triangle-3.1} is at most $\delta'$ and this complete the proof of \Cref{eq:prg-lc-triangle-3}. Finally, note that \Cref{eq:prg-lc-triangle} follows from \Cref{eq:prg-lc-triangle-1}, \Cref{eq:prg-lc-triangle-2} and \Cref{eq:prg-lc-triangle-3}. This completes the proof of \Cref{lem:prg-lc-claim-1}.
	    \end{proof}
	    
	    Given \Cref{lem:prg-lc-claim-1}, \Cref{thm:prg-lc} now follows by letting $i=t$.
	\end{proof}

\end{document}